\documentclass[11pt,letterpaper]{article}

\usepackage[T1]{fontenc}
\usepackage{lmodern}
\usepackage[margin=1in]{geometry}
\usepackage{amsmath,amssymb,mathtools}
\usepackage{dsfont}
\usepackage{amsthm}
\usepackage{microtype}
\usepackage{booktabs}
\usepackage{array}
\usepackage{enumitem}
\usepackage{xcolor}
\usepackage{url}
\usepackage{hyperref}
\newcommand{\doi}[1]{\href{https://doi.org/#1}{\nolinkurl{#1}}}
\newtheorem{theorem}{Theorem}
\newtheorem{lemma}{Lemma}
\newtheorem{proposition}{Proposition}
\newtheorem{corollary}{Corollary}
\theoremstyle{definition}
\newtheorem{definition}{Definition}
\theoremstyle{remark}
\newtheorem{remark}{Remark}

\usepackage[nameinlink,capitalise,noabbrev]{cleveref}
\usepackage{xspace}

\hypersetup{
  colorlinks=true,
  linkcolor=blue!45!black,
  citecolor=green!35!black,
  urlcolor=blue!55!black,
  pdftitle={Leakage Resilience Below Half: Shamir Secret Sharing and Generalized Reed-Solomon Codes},
  pdfauthor={Anonymous Authors},
  pdfsubject={STOC submission draft},
  pdfkeywords={Shamir secret sharing, local leakage resilience, Reed--Solomon codes, list decoding, Gowers norms}
}

\setlist[itemize]{leftmargin=1.35em,itemsep=0.25em,topsep=0.35em}
\setlist[enumerate]{leftmargin=1.55em,itemsep=0.25em,topsep=0.35em}

\newcommand{\Fp}{\mathbb F_p}
\newcommand{\C}{\mathbb C}

\newcommand{\E}{\mathbb E}
\newcommand{\T}{\mathsf{T}}
\renewcommand{\vec}[1]{\boldsymbol{#1}}
\newcommand{\Prb}{\mathop{\rm Pr}}
\newcommand{\SD}{\mathop{\rm SD}}
\newcommand{\wt}{\mathop{\rm wt}}
\newcommand{\supp}{\mathop{\rm supp}}
\newcommand{\distH}{d_{\rm H}}
\newcommand{\Share}{\mathsf{share}}

\newcommand{\GRS}{\mathsf{GRS}}
\newcommand{\List}{\mathsf{List}}
\newcommand{\poly}{\mathsf{poly}}
\newcommand{\eps}{\varepsilon}
\newcommand{\zetaGap}{\zeta}
\newcommand{\one}{\mathds{1}}

\newcommand{\abs}[1]{\left|#1\right|}
\newcommand{\norm}[1]{\left\lVert#1\right\rVert}

\newcommand{\floor}[1]{\left\lfloor#1\right\rfloor}
\newcommand{\ceil}[1]{\left\lceil#1\right\rceil}
\newcommand{\comp}{\mathcal C}

\crefname{equation}{Eq.}{Eqs.}
\Crefname{equation}{Equation}{Equations}

\title{Local Leakage Resilience of Shamir Secret Sharing:\\ Below Half Threshold Rate}
\author{Hai H. Nguyen\thanks{\textcolor{black}{\texttt{hai.h.nguyen@nus.edu.sg}}}}
\date{}

\begin{document}
\pagenumbering{gobble}
\hypersetup{pageanchor=false}
\maketitle

\begin{abstract}
We study the local leakage resilience of Shamir secret sharing, where
an adversary leaks a bounded number of bits from every share. For \(n\)
parties and reconstruction threshold \(k\) over a prime field \(\Fp\)
with \(p>n\), we prove the following.
\begin{itemize}
  \item There is an absolute constant \(c_{\mathrm{bh}}<1/2\) such
    that, for \(k\ge c_{\mathrm{bh}}n\), Shamir sharing is resilient
    to arbitrary one-bit leakage from every share. The best previous
    threshold was \(0.668n\).
  \item For every fixed \(0<\delta<1/2\) and \(k\ge(1/2+\delta)n\),
    Shamir sharing is resilient to \(\floor{\frac{\delta}{24}\log_2n}\) bits
    of arbitrary local leakage per share. By a recent attack, this
    number of bits is optimal up to a constant factor depending on
    \(\delta\).
  \item For every fixed threshold rate \(c>0\) and \(k\ge cn\), Shamir
    sharing is resilient to quadratic-residue leakage and to almost all
    local leakage functions of up to \(\eta_c\log_2p\) bits, for a
    constant \(\eta_c>0\).
\end{itemize}
% All results hold with negligible error for all sufficiently large \(n\),
% with no upper bound on \(p\).

All three results follow from one analytic framework, which we develop
for generalized Reed--Solomon codes: local leakage from a uniform
codeword is statistically close to the same leakage from independent
uniform field elements. The framework expresses the leakage distance
as the norm of a function on the leakage outputs. The Efron--Stein
decomposition splits this function into orthogonal components indexed
by supports of dual codewords. Two tools bound the components: list
recovery of the dual code, which is ordinary list decoding for one-bit
leakage, and a generalized von Neumann inequality for MDS systems of
linear forms that retains linearly many Gowers norms. The first tool
gives the result above half, the second gives the constant-rate
result, and applying both to different ranges of supports in the same
decomposition gives the result below half.
\end{abstract}

\clearpage
\setcounter{tocdepth}{2}
\tableofcontents
\clearpage
\pagenumbering{arabic}
\hypersetup{pageanchor=true}
\section{Introduction}
\label{sec:introduction}

Shamir secret sharing~\cite{Shamir79} is a fundamental cryptographic
primitive with many applications. In an \((n,k)\)-Shamir scheme, a secret is encoded in
\(n\) evaluations, called shares, of a random polynomial of degree below \(k\) whose
constant coefficient is the secret. Any \(k\) shares determine the
secret, and any fewer reveal nothing about it. 
% We call \(k/n\) the
% \emph{threshold rate}, or simply the \emph{rate}. 
The algebraic
structure also allows efficient computation on shared secrets, which makes the
scheme a central building block of secure multiparty computation
(MPC)~\cite{BGW88,CCD88}. Classical MPC protocols tolerate passive
corruption of fewer than half the parties, or active corruption of
fewer than one third, with reconstruction thresholds near \(n/2\) and
\(n/3\) respectively.

What happens if, instead of learning a few shares in full, an
adversary learns a little about \emph{every} share? The question is
motivated by side-channel attacks, which extract information from
physical implementations through observations such as running time
and power consumption~\cite{Kocher96,KocherJaffeJun99}. Benhamouda,
Degwekar, Ishai, and Rabin~\cite{BDIR18} formalized \emph{locally
leakage-resilient secret sharing}.\footnote{The notion also appears
implicitly in the work of Goyal and Kumar~\cite{GK18}.} In this
model the adversary obtains a bounded number of bits from each share
through a function of its choice, and local leakage resilience
requires the joint leakage to be statistically independent of the
secret. The guarantee is strong because the leakage functions are
arbitrary: even one-bit leakage permits any Boolean function of each
share, not merely a bit of its representation.

% The difficulty is in the collective effect of these observations.
% For each fixed secret, any \(k-1\) Shamir shares are independent
% uniform field elements. Thus any \(k-1\) leakage outputs already
% have exactly the same joint distribution for every secret. Any
% distinguishing information in the full collection must come from
% dependencies involving more coordinates. 
% Whatever distinguishes secrets must come from dependencies among
% \(k\) or more coordinates, each seen through one bit.

Local leakage can be surprisingly powerful. For suitable parameters
over binary extension fields, Reed--Solomon repair
algorithms~\cite{GuruswamiWootters16} even recover the entire secret
from one bit of leakage per share. Over prime fields, Benhamouda et
al.~\cite{BDIR18} proved resilience to arbitrary one-bit local
leakage when \(k\ge0.907n\), and substantial effort has gone into
lowering this threshold: independent improvements gave
\(0.8675n\)~\cite{MPSW21} and \(0.85n\)~\cite{BDIR21}, followed by
\(0.78n\)~\cite{MNPW22}, \(0.69n\)~\cite{KleinKomargodski23}, and
\(0.668n\)~\cite{Kasser24}.\footnote{The last two assume a field of
size \(2^{o(n)}\).} 
Below half, resilience is known only for restricted leakage
classes: physical-bit leakage with high probability over randomly
chosen evaluation points in sufficiently large prime
fields~\cite{MNPSW21}, quadratic-residue and almost all one-bit local
leakage~\cite{Nguyen24}, and sufficiently unbalanced
leakage~\cite{KleinKomargodski23}. Previous guarantees against arbitrary
one-bit local leakage all require threshold rates strictly above half.
% The central question is therefore:
% \begin{quote}
% \itshape
% Is Shamir sharing resilient to arbitrary one-bit local leakage when
% the reconstruction threshold reaches, or falls below, half of
% the parties?
% \end{quote}

Benhamouda et al.~\cite{BDIR18,BDIR21} conjectured that Shamir sharing over prime fields is resilient to arbitrary one-bit local leakage at every positive constant rate.\footnote{Their conjecture applies more generally to secret-sharing schemes based on MDS codes.} Despite the subsequent improvements, this conjecture remained open even at rate half.
There is also a technical obstacle. Maji et al.~\cite{MPSW21}
showed that the bound obtained by summing Fourier magnitudes in
earlier analyses cannot establish resilience below rate half,
already for quadratic-residue leakage.\footnote{For a share \(x\in\Fp\),
this one-bit leakage outputs \(1\) if \(x=y^2\) for some \(y\in\Fp\),
and \(0\) otherwise.}

\emph{Summary of our results.}
We prove that $(n,k)$-Shamir sharing is resilient whenever \(k\ge c_{\mathrm{bh}}n\), for an absolute constant \(c_{\mathrm{bh}}<1/2\), establishing the conjecture in a range that extends below half. We also obtain logarithmic-bit leakage resilience above half (by any fixed constant), and guarantees against restricted leakage families at every positive constant rate.

% We answer this question affirmatively, proving security whenever
% \(k\ge c_{\mathrm{bh}}n\) for an absolute constant
% \(c_{\mathrm{bh}}<1/2\).
% We also address two complementary questions:
% how many bits of leakage per share Shamir sharing withstands above
% half, and which leakage families it withstands at every constant
% threshold rate.
% The next subsection provides the precise statements.
% Our approach
% establishes a pseudorandomness property of generalized
% Reed--Solomon codes under local observations, from which the
% secret-sharing guarantees follow. The results below give the
% precise bounds in these three regimes.

\subsection{Our contributions}
\label{sec:contributions}

We first introduce notation.
Consider $(n,k)$-Shamir sharing with $n$ parties, reconstruction threshold $k$, and distinct nonzero evaluation points  \(\vec \alpha = (\alpha_1,\ldots,\alpha_n)\) over a prime field $\Fp$, where $p >n$.
To share a secret $s\in\Fp$, this scheme samples a uniformly random polynomial $P \in \Fp[X]$ of degree less than $k$ such that $P(0) = s$ and gives party $i$ the share $P(\alpha_i)$.
An \(m\)-bit local leakage function is a tuple of maps $\vec \tau = (\tau_1, \tau_2, \ldots, \tau_n)$, where \(\tau_i\colon\Fp\to\{0,1\}^m\).
For \(\vec x=(x_1,\ldots,x_n)\in\Fp^n\), write
\[
  \vec\tau(\vec x):=(\tau_1(x_1),\ldots,\tau_n(x_n)).
\]
The \emph{leakage distribution} induced by \(\vec\tau\) on secret
\(s\) is the distribution of
\(
  (\tau_1(P(\alpha_1)),\ldots, \tau_n(P(\alpha_n))\bigr),
\)
where \(P\) is sampled as above.
The scheme is \(\eps\)-resilient to $m$-bit local leakage if, for every such
$m$-bit local leakage function, the leakage distributions of any two secrets have
statistical distance at most \(\eps\). 
We call \(k/n\) the \emph{threshold rate} and \(m/\log_2p\) the
\emph{leakage rate}, the fraction of each share's bit length revealed
by the leakage.
All following results hold for every
choice of evaluation points and all sufficiently large \(n\), with
constants and the required lower bound on \(n\) depending only on the
fixed parameters. 
In particular, there is no upper bound on \(p\).

\paragraph{Below half.} Our first result crosses the $1/2$ rate for one-bit local leakage.
\begin{theorem}[Below half]
\label{thm:intro-main}
There is an absolute constant \(c_{\mathrm{bh}}<1/2\) such that every
\((n,k)\)-Shamir scheme with \(k\ge c_{\mathrm{bh}}n\) is
\(\exp(-\Omega(n))\)-resilient to one-bit local leakage.
\end{theorem}

The previous best bound for arbitrary one-bit leakage
was \(0.668n\), under the field-size restriction
\(p=2^{o(n)}\)~\cite{Kasser24}. Our theorem crosses half without
this restriction.\footnote{This allows a \(\lambda\)-bit prime \(p\)
with \(n=O(\lambda)\), a common parameter regime in MPC.}
% Earlier below-half guarantees for Shamir sharing
% concerned restricted leakage classes: physical bits with randomly
% chosen evaluation points~\cite{MNPSW21}, quadratic-residue and almost
% all Boolean leakage~\cite{Nguyen24}, and sufficiently unbalanced
% leakage~\cite{KleinKomargodski23}. 
% Below-half results are also known
% for ramp sharing~\cite{DoronEtAl26}, which allows a gap between the
% privacy and reconstruction thresholds.
In particular, the theorem covers \(k=\ceil{n/2}\), the reconstruction
threshold at the classical \(n/2\) passive-corruption bound in MPC.
The proof gives the explicit, unoptimized constant
\(c_{\mathrm{bh}}=1/2-10^{-14}\) (\cref{cor:explicit-below-half}).
The full BDIR conjecture remains open, in particular at threshold
rate \(1/3\). At smaller thresholds, Nielsen and Simkin~\cite{NielsenSimkin20} give
an inefficient one-bit leakage attack when \(k\) is of order $n /\log p$, which is
\(n/\log n\) when \(p=\poly(n)\).
Closing this gap remains an open problem.

\paragraph{Asymptotically optimal leakage above half.}
Progress on the reconstruction threshold has concentrated on one-bit
leakage. A second basic question is how many bits per share Shamir
sharing withstands. We show that, at every threshold rate above half,
the answer is logarithmic in \(n\): for polynomial-size fields
\(p=\poly(n)\), this is a constant fraction of each share's bit length.

\begin{theorem}[Logarithmic leakage above half]
\label{thm:intro-optimal}
Fix \(0<\delta<1/2\). Any \((n,k)\)-Shamir sharing with
\(k\ge(1/2+\delta)n\) is
\(\exp(-\Omega_\delta(n^{1-\delta/12}))\)-resilient to
\(\floor{(\delta/24)\log_2 n}\)-bit local leakage.
\end{theorem}

Previous guarantees of a constant leakage rate required a threshold
rate tending to one. Benhamouda et al.~\cite[Corollary~4.13]{BDIR21}
obtain it for \(n<p\le 2n\) and \(k>n-n^{1/4}\). Our threshold rate stays
fixed above half, and the field may be arbitrarily large.

The number of leaked bits is tight up to a constant factor depending on
\(\delta\). Boudgoust and Simkin~\cite{BoudgoustSimkin26} give an
efficient attack using \(\log_2 k+O(1)\) bits per share over arbitrary
finite fields.

The proof reduces \(m\)-bit leakage resilience to list recovery of the
dual code with input lists of size \(2^m-1\); for one bit this is
ordinary list decoding. Letting \(m\) grow requires the dependence of
the list size on the input-list size, which we extract from the proofs
of Jeronimo~\cite{Jeronimo26} and of Dao, Kominers, and Thaler
(DKT)~\cite{DKT26} in \cref{app:growing-lists,app:dkt-lists}.

The criterion of Maji, Paskin-Cherniavsky, Suad, and Wang~\cite[Technical
Lemma~1]{MPSW21} is also a list-recovery condition on the dual code, and
combined with the same bound it yields \cref{thm:intro-optimal} for
every fixed \(m\). Its parameters, however, degrade doubly exponentially
in \(m\): its input lists have size \(2^m T^2\) with \(T=2^{\Theta(m/\sigma)}\)
and \(\sigma=O(4^{-m})\), because the underlying \(\ell_1\) bound sums
over all \(2^{mn}\) leakage outputs. This direct application gives only
\(O(\log\log n)\) bits.

\paragraph{Restricted leakage families at every constant threshold rate.}
Our third result concerns every positive constant threshold rate,
including below half, for two restricted leakage families:
quadratic-residue leakage, where party \(i\) reveals whether
\(a_iP(\alpha_i)+b_i\) is a square, for arbitrary \(a_i\in\Fp^*\) and
\(b_i\in\Fp\); and almost all tuples of local leakage functions
revealing \(\Theta(\log p)\) bits per share.

\begin{theorem}[Structured leakage at constant threshold rate]
\label{thm:intro-structured}
Fix \(0<c<1\). There are constants \(\eta_c,a_c,b_c>0\) such that
any \((n,k)\)-Shamir sharing with \(k\ge cn\) is
\begin{enumerate}
  \item \(\exp(-a_cn\log p)\)-resilient to quadratic-residue
    leakage (with arbitrary affine shifts), and
  \item for every integer \(1\le m\le\eta_c\log_2p\),
    \(\exp(-a_cn\log p)\)-resilient to all but an
    \(\exp(-b_c\sqrt p)\) fraction of \(m\)-bit local leakage functions.
\end{enumerate}
\end{theorem}

Nguyen~\cite{Nguyen24} proved resilience to quadratic-residue and
random one-bit leakage using higher-order Fourier analysis, but the
bounds require a field size exponential in \(n\) at constant threshold
rate. Unlike
\cref{thm:intro-optimal}, the leakage rate in part~2 stays constant
for every prime \(p>n\), not only polynomial-size fields, and below
half. For superpolynomial-size fields, this separates worst-case from
average-case leakage: the attack of Boudgoust and
Simkin~\cite{BoudgoustSimkin26} limits the worst-case leakage rate to
\(O(\log k/\log p)=o(1)\).

At smaller thresholds, quadratic-residue leakage can be insecure:
for \(n=p-1\) and \(k=O(\sqrt n)\), Kopparty's algorithm
for recovering a polynomial from noisy quadratic-character
values~\cite{Kopparty26} yields an attack with the same shifted
quadratic-residue leakage (\cref{sec:sqrt-threshold-field}).

\paragraph{From Shamir sharing to codes.}
For distinct evaluation points \(\vec\alpha\in\Fp^n\) and nonzero
multipliers \(\vec v\in(\Fp^*)^n\), the generalized Reed--Solomon
code of dimension \(d\) is
\[
  \GRS_{\vec\alpha,\vec v}(n,d)
  =\bigl\{(v_iQ(\alpha_i))_{i=1}^n:\ Q\in\Fp[X],\ \deg Q<d\bigr\}.
\]
The shares of the zero secret in an $(n,k)$-Shamir sharing are uniform on
\(C_0=\GRS_{\vec\alpha,\vec\alpha}(n,k-1)\), and the shares of secret
\(s\) are uniform on the translate \(s\vec1+C_0\), where
\(\vec1=(1,\ldots,1)\). Call a linear code
\(C\subseteq\Fp^n\) \(\eps\)-resilient to \(m\)-bit local leakage if,
under every tuple of \(m\)-bit leakage maps, a uniform codeword and a
vector of \(n\) independent uniform field elements have leakage
distributions within statistical distance \(\eps\). We prove all
three results in this code form
(\cref{thm:grs-below-half,thm:grs-constant-bits,thm:structured-leakage}),
as a pseudorandomness property of GRS codes under local leakage, and
the guarantees for Shamir sharing follow: for arbitrary \(m\)-bit
local leakage, resilience of the code implies resilience of the scheme
with twice the error, and conversely with error larger by a factor \(O(n)\)
(\cref{lem:shamir-from-code}).

\paragraph{Technical Highlights.}
We develop an analytical framework in which all three theorems are
proved. It combines the Efron--Stein orthogonal decomposition over the
product space of leakage outputs with Fourier analysis over the
uniform measure on the inputs. Its key feature is flexibility: the
framework organizes the leakage analysis into orthogonal components,
and each component
can be analyzed either through its \(L^2\)-norm using linear Fourier
analysis, or through its \(L^1\)-norm using higher-order Fourier
analysis. Previous analyses each took a single route. Benhamouda et
al.~\cite{BDIR18, BDIR21} and their successors~\cite{MNPSW21, MPSW21} bound an \(\ell_1\) sum of Fourier magnitudes over the dual code; Klein and
Komargodski~\cite{KleinKomargodski23, Kasser24} bound an \(\ell_2\) quantity
obtained from Pinsker's inequality; Nguyen~\cite{Nguyen24} bounds
Gowers norms of the leakage functions. The framework brings these
analytic approaches together, allowing different estimates on different
components. Their combination within the same decomposition is what
crosses half.
% \Cref{sec:overview} gives the details.

\paragraph{Concurrent and independent work.}
Sun and Wootters~\cite{SunWootters26Leakage} independently prove
one-bit local leakage resilience of Shamir sharing at a constant
threshold rate below half, obtaining the same qualitative
guarantee as \cref{thm:intro-main}. They also prove resilience to
any fixed number of leaked bits per share above half.
The approaches
share an important feature: both use list recovery of the dual code
for small supports and Gowers-norm estimates for large supports.
Our work additionally establishes logarithmic-bit resilience above half
and the restricted-leakage guarantees at every positive constant
threshold rate (\cref{thm:intro-optimal,thm:intro-structured}).
Sun and Wootters also obtain existential and quantum-algorithmic results
for worst-case optimal polynomial intersection, which we do not
study here.

\subsection{Technical overview}
\label{sec:overview}

By the reduction described above (\cref{lem:shamir-from-code}), it suffices to prove the code form
of our results: local leakage cannot distinguish a uniform codeword
from independent uniform field elements. Fix a GRS code
\(C=\GRS_{\vec\alpha,\vec v}(n,k-1)\) and leakage maps
\(\tau_i\colon\Fp\to\{0,1\}^m\). Let \(\vec U_C\) and
\(\vec U_{\Fp^n}\) be uniformly distributed on \(C\) and \(\Fp^n\),
respectively. The goal is to show that the statistical distance between \(\vec\tau(\vec U_C)\)
and \(\vec\tau(\vec U_{\Fp^n})\), which is called \emph{leakage distance}, is small. 

\paragraph{The framework.}
Under independent uniform inputs, the leakage outputs are independent. 
Thus \(\vec\tau(\vec U_{\Fp^n})\) has a product distribution
\(\nu=\prod_i\nu_i\), where \(\nu_i\) is the distribution of
\(\tau_i(U_{\Fp})\).
This allows us to use Fourier analysis on general product spaces (\cite[Chapter~8]{ODonnell14}). In particular, the
Efron--Stein decomposition is a standard way to group a Fourier
expansion into orthogonal components indexed by coordinate sets,
without choosing an explicit basis. We first express the leakage
distance as a norm under \(\nu\), then apply this decomposition
and use the code's linear structure to identify the components.
\Cref{sec:decomposition} develops this approach in the following steps.

\emph{Step 1: Write the leakage distance as a norm.}
For each leakage output \(\vec b = (b_1, b_2, \ldots, b_n)\), we compare its probability under leakage
from a uniform codeword with \(\nu(\vec b)\). To express this ratio,
start with the corresponding ratio on the inputs:
\[
  f_C(\vec x)=\frac{\Prb[\vec U_C=\vec x]}{\Prb[\vec U_{\Fp^n}=\vec x]}
  =\frac{p^n}{|C|}\one[\vec x\in C],
\]
whose mean under uniform input is one. Leakage forgets which input
produced \(\vec b\), so the output ratio is the average of this input
ratio over all compatible inputs. This naturally gives the Markov operator
(\cref{def:leakage-operator})
\[
  (\T f)(\vec b)=\E\bigl[f(\vec U_{\Fp^n})\bigm|\vec\tau(\vec U_{\Fp^n})=\vec b\bigr].
\]
Its one-coordinate version is
\((\T_if)(b)=\E[f(U_{\Fp})\mid\tau_i(U_{\Fp})=b]\), and
independence gives \(\T=\bigotimes_i\T_i\). Thus, we have
(\cref{lem:leakage-distance})
\[
  (\T f_C)(\vec b)
  =\frac{\Prb[\vec\tau(\vec U_C)=\vec b]}{\nu(\vec b)}.
\]
The leakage distance is therefore the average deviation of this
ratio from one:
\[
  \SD\bigl(\vec\tau(\vec U_C),\vec\tau(\vec U_{\Fp^n})\bigr)
  =\tfrac12\E_{\vec b\sim\nu}\bigl|(\T f_C)(\vec b)-1\bigr|
  =\tfrac12\|\T f_C-1\|_{1,\nu}.
\]
Thus we need to bound how far \(\T f_C\) is from the constant
function \(1\), averaging over the product distribution \(\nu\).
Here and below, we only consider outputs with \(\nu(\vec b)>0\).

\emph{Step 2: Separate the contributions of different coordinate sets.}
The Efron--Stein decomposition separates a function on a product space into
orthogonal components indexed by coordinate sets \(A \subseteq [n]\)
(\cref{def:efron-stein}). A component depends only on \(A\) and
disappears when any coordinate in \(A\) is averaged out: it records
an interaction requiring all those coordinates.
Apply this decomposition to \(\T f_C\) under the leakage distribution
\(\nu\). Its constant component ($A = \emptyset$) is one, so \(\T f_C-1\) is the sum of its nonconstant
components. To identify them, we can instead decompose \(f_C\)
under the uniform distribution on \(\Fp^n\) and then apply \(\T\).
Because leakage acts separately on each coordinate, the two operations commute
(\cref{lem:product-expectation}):
\[
  (\T f_C)^{=A}=\T(f_C^{=A}),
  \qquad
  \T f_C-1=\sum_{\emptyset\ne A\subseteq[n]}\T(f_C^{=A}).
\]
For all distinct coordinate sets \(A\) and \(B\), the corresponding
output components are orthogonal under \(\nu\). We can therefore
analyze the dependence on each coordinate set separately.

\emph{Step 3: Identify the components using the dual code.}
To make these components explicit, use the linear structure of the
code. Fourier characters detect its linear constraints: the character
\(\chi_{\vec y}(\vec x)=\prod_i\chi_{y_i}(x_i)\), where
\(\chi_a(x)=\exp(2\pi iax/p)\), is identically one on \(C\) when
\(\vec y\in C^\perp\), and otherwise averages to zero over \(C\).
Thus the Fourier coefficients of \(f_C\) are one on \(C^\perp\)
and zero elsewhere. Grouping these characters by their
nonzero coordinates gives exactly the Efron--Stein components
(\cref{lem:dual-code-formula,lem:fourier-components}):
\[
  f_C^{=A}=\sum_{\vec y\in C^\perp:\ \supp(\vec y)=A}\chi_{\vec y}.
\]
Since \(C^\perp\) has minimum distance \(k\), all components with
\(0<|A|<k\) vanish. We therefore only need to bound components with
\(|A|\ge k\).

\emph{Step 4: Choose how to bound the components.}
The two tools below provide different kinds of estimates: list
recovery bounds squared \(L^2\) norms, while Gowers norms bound
\(L^1\) norms. We need not choose one tool for every support.
Using \(\|\cdot\|_{1,\nu}\le\|\cdot\|_{2,\nu}\) and orthogonality, we can
bound a group of components by the square root of the sum of their
squared \(L^2\) norms. For the others, the triangle inequality lets
us add individual \(L^1\) bounds. Choosing a cutoff
\(k-1\le r_0\le n\) gives (\cref{prop:leakage-decomposition})
\[
  \SD\bigl(\vec\tau(\vec U_C),\vec\tau(\vec U_{\Fp^n})\bigr)
  \le\frac12\Bigl(\sum_{k\le|A|\le r_0}\|\T(f_C^{=A})\|_{2,\nu}^2\Bigr)^{1/2}
   +\frac12\sum_{|A|>r_0}\|\T(f_C^{=A})\|_{1,\nu} .
\]
Above half we take \(r_0=n\), using only the first
part. For the restricted leakage families we take \(r_0=k-1\), using
only the second. Below half we choose an intermediate cutoff
close to \(n\) and combine both tools.

\paragraph{Above half: \(L^2\) norms and list recovery of the dual code.}
Here we take \(r_0=n\) and then bound the sum of squared component norms.
The first question is how the leakage acts on a component, which is a
sum of characters \(\chi_{\vec y}\) over dual words with support
\(A\). Applying \(\T\) to every coordinate loses the orthogonality
among them, but applying it to only \(k-1\) coordinates \(I\subseteq A\)
does not: two dual words with support \(A\) that agree outside \(I\)
differ on at most \(k-1\) coordinates, hence coincide, so their
characters on \(A\setminus I\) are distinct and still orthogonal; and
applying the remaining \(\T_i\) afterwards can only shrink an \(L^2\)
norm by the contraction property of Markov operators. This gives (\cref{lem:component-bound})
\[
  \|\T(f_C^{=A})\|_{2,\nu}^2
  \le\sum_{\vec y\in C^\perp:\ \supp(\vec y)=A}
     \prod_{i\in I}\|\T_i\chi_{y_i}\|_{2,\nu_i}^2 .
\]
So each dual word is charged \(k-1\) one-coordinate factors
\(\|\T_i\chi_a\|_{2,\nu_i}^2\). There are two facts about them
(\cref{lem:local-fourier-energy}). Their sum over \(a\ne0\) is at most
\(Q-1\), where \(Q=2^m\). And each is at most
\(\eta_Q<1\). 
% For one bit,
% \(\eta_2=1/2\) by conjugate symmetry, with no condition on \(p\)
% beyond oddness.

Now comes the counting problem: a sum over all dual words with support
\(A\), summed over \(A\), of products of \(k-1\) factors. Spending all
\(k-1\) factors on decay leaves a bare count of dual words, which is
\(p^{\,n-k+1}\) and useless. The trick is to split the factors. Keep
\(j\) of them, bound the other \(k-1-j\) by \(\eta_Q\), and look at
the sum that remains. It is linear in the weights
\(w_i(a)=\|\T_i\chi_a\|_{2,\nu_i}^2\) of each coordinate separately, with
nonnegative coefficients, under the constraints \(w_i(a)\in[0,1]\) and
\(\sum_{a\ne0}w_i(a)\le Q-1\). Such a sum is maximized at an extreme point,
which puts weight one on \(Q-1\) frequencies and zero elsewhere
(\cref{lem:weights-to-lists}). After this replacement the sum counts
dual words that agree with lists \(V_i\subseteq\Fp^*\) of size \(Q-1\)
on at least \(j\) coordinates (\cref{lem:component-list-recovery}).
This count is bounded by the list size at radius \(n-j\), in the notation
\(\List_{C^\perp}(R,\ell)\) (\cref{def:list-size}):
\[
  \sum_{|A|\ge k}\|\T(f_C^{=A})\|_{2,\nu}^2
  \le\eta_Q^{\,k-1-j}\,\List_{C^\perp}(n-j,Q-1).
\]
The lists live on Fourier frequencies.
For one bit they are singletons and this is list decoding.

The dual of a GRS code is a GRS code, so either of the
list-recovery results of Jeronimo and DKT~\cite{Jeronimo26,DKT26} makes
\(\List_{C^\perp}(R,\ell)\) polynomial in \(n\) for fixed input-list
size \(\ell\) and fixed gap \(\zeta>0\), at radius
\(R\le k-1-\zeta n\) (\cref{thm:rs-list-decoding}).
With \(j=n-R\) there are at least \(2R-n\)
factors left for decay, and the bound is
\(\frac12\eta_Q^{R-n/2}\sqrt{\List_{C^\perp}(R,Q-1)}\)
(\cref{thm:list-recovery-transfer}). This also says exactly where the
argument stops: it needs \(R>n/2\) with \(R<k\), hence \(k>n/2\).
Below half no choice of \(j\le k-1\) reaches a radius covered by
this polynomial-list bound. Above \(1/2+\delta\), a linear number of
decay factors remain, so their product absorbs the polynomial list
size. Increasing the leakage alphabet has two costs: the local decay
becomes weaker, and the list-recovery bound becomes larger. The same
linear number of factors can absorb both costs as long as the
logarithm of the list bound stays smaller than the decay exponent.
Keeping the dependence on \(Q\) explicit in these proofs
(\cref{app:growing-lists,app:dkt-lists}) allows \(Q\) to grow to
\(n^{\delta/24}\), giving the logarithmic leakage in
\cref{thm:intro-optimal}.

\paragraph{Every constant threshold rate: \(L^1\) norms and iterated Cauchy--Schwarz.}
For the restricted families, we take \(r_0=k-1\) and bound each
component in \(L^1\). By a simple calculation, the
\(L^1\) norm of each component is a sum of absolute values of averages (\cref{lem:output-indicator-components}),
each of the following form:
\[
  \Lambda_\Psi(\vec f)
  =\E_{\vec z}
       \left[\prod_{i\in A}f_i(L_i(\vec z))\right].
\]
Here \(\vec z\in\Fp^{k-1}\) is the vector of independently and
uniformly chosen polynomial coefficients, and \(L_i(\vec z)\) is
the \(i\)-th coordinate of the resulting codeword. We write
\(\Psi=(L_i)_{i\in A}\), where \(|A|\ge k\). The functions \(f_i\)
are bounded by one and determined by the leakage maps and outputs
(see \cref{sec:iterated-cs}).

Such averages arise naturally in additive combinatorics. For indicator
functions, they give normalized counts of configurations described
by linear forms, such as arithmetic progressions~\cite{TaoVu06,GreenTao10}.
Repeated applications of Cauchy--Schwarz bound these averages in terms
of Gowers norms~\cite{Gowers01,GreenTao10}.

Nguyen~\cite{Nguyen24} uses the scalar generalized von Neumann inequality
to bound such averages. For a suitable Gowers order \(t\), it gives
\[
  |\Lambda_\Psi(\vec f)|
  \le\min_{i\in A}\|f_i\|_{U^t}.
\]
This retains only one small factor to compensate for exponentially
many averages, leading to the large field sizes in that analysis.
Nguyen~\cite[Remark~1]{Nguyen24} conjectured that a better understanding
of this inequality for the linear forms of MDS codes and the leakage
functions arising here could lead to smaller-field bounds.

We observe that Manners's formulation~\cite{Manners21}
allows us to retain a product of Gowers norms. The idea is to group
several independent forms into one vector-valued map. The Gowers norm
of the corresponding product function then factors into the individual
Gowers norms. The coordinate forms of a GRS code are MDS: any \(k-1\)
of them are independent. This structure gives, at every constant
threshold rate \(c>0\), a constant Gowers order \(t\) and a set
\(J\subseteq A\) such that
\[
  |\Lambda_\Psi(\vec f)|
  \le\prod_{i\in J}\|f_i\|_{U^t},
  \qquad |J|=\Omega_c(n).
\]
See \cref{lem:iterated-cs,cor:iterated-cs-count}.
The key saving comes from multiplying linearly many factors: even
when the retained Gowers norms are close to one, a fixed gap below
one gives \(\exp(-\Omega(n))\) decay for each average. This will be
crucial for the result below half, where the norms need not
decrease with \(p\).

For quadratic-residue leakage and almost all local leakage maps,
the local Gowers norms are at most \(p^{-\Omega_c(1)}\), by character-sum
estimates and concentration, respectively
(\cref{sec:quadratic-leakage,sec:random-bit-leakage}).
The product bound therefore gives the stronger
\(\exp(-\Omega_c(n\log p))\) decay. This absorbs the sum over
supports and leakage outputs even when the number of leaked bits is
a sufficiently small constant multiple of \(\log p\). Hence \(p>n\)
suffices, giving \cref{thm:intro-structured}; the parameter calculation
is in \cref{prop:gowers-leakage-transfer}.

\emph{Why not arbitrary leakage at every constant rate?}
For arbitrary one-bit leakage, the Gowers norms need not be small
enough for the resulting exponential decay to outweigh the
exponentially many terms. This becomes more restrictive at smaller
threshold rates, where the bound retains fewer Gowers-norm factors.

\paragraph{Below half: both tools, and a uniform bound for one bit.}
We combine the two tools to handle arbitrary one-bit leakage below
half, choosing a cutoff \(r_0\) close to \(n\), using \(L^2\) bounds
for the smaller supports and \(L^1\) bounds for the larger ones.

\emph{Smaller supports.} The dual words contributing to these
components have many zero coordinates. These zeros provide additional
agreements with the list \(\{0\}\) when applying list decoding.
Fewer agreements must therefore come from the Fourier factors,
leaving more factors available to provide decay. A modification of
the earlier counting argument uses these extra agreements and bounds
all supports of a given size together
(\cref{lem:fixed-support-list-recovery,lem:list-decoding-energy}).
The resulting estimate makes their total squared \(L^2\) contribution
exponentially small for supports that omit a fixed fraction of the
coordinates, at threshold rates sufficiently close to half
(\cref{lem:elementary-low-band}).

\emph{Larger supports.} These supports have too few zero coordinates
for the counting argument, so we turn to the linear-form averages. 
For one-bit leakage, the \(L^1\) norm of each component is the
absolute value of a single such average (\cref{lem:binary-components}).
Its factors \(f_i\) are obtained by representing each leaked bit as
\(\pm1\) and subtracting its mean under uniform input. Using Fourier
symmetry and the fact that these functions take at most two values, we
prove that
\[
  \|f_i\|_{U^3}\le\kappa_0<1
\]
for an absolute constant \(\kappa_0\)
(\cref{thm:qual-local-bounds}). The product bound retains linearly many
such factors even for full support at rate half
(\cref{cor:iterated-cs-count}).

The Gowers-norm bound controls only some of the factors in the average.
The other factors are bounded by their maximum absolute values, which
can be as large as two. Multiplying these bounds can therefore offset
the saving from the Gowers norms.
When the total squared bias is small, this loss is controlled;
otherwise, the smaller variances give exponential decay by
Cauchy--Schwarz and MDS independence, following a principle similar to that in
\cite[Section~6]{KleinKomargodski23}.
Together, these estimates make each large-support contribution
exponentially small, uniformly over the leakage maps
(\cref{lem:qual-padded-correlation}). Taking the cutoff sufficiently
close to \(n\) makes the number of remaining supports small enough
for this decay to survive summation (\cref{prop:qual-high-sum}).

\emph{Putting it together.} At the limiting parameters of rate half
and zero decoding gap, both estimates leave a strict exponential
margin. By continuity, this margin survives at a fixed rate slightly
below half and a small positive decoding gap, with both ranges
of supports still covered. The list-decoding theorem is used only at
that positive gap. This proves \cref{thm:intro-main};
\cref{sec:splice} gives the parameter choices.

\subsection{Related work}

\paragraph{New constructions of leakage-resilient secret sharing.}
Following~\cite{GK18,BDIR18}, a line of work has developed secret-sharing
schemes designed to withstand local leakage. General compilers support
arbitrary monotone access structures while providing strong leakage
guarantees with controlled share-size overhead~\cite{SV19,AggarwalEtAl19}.
Other constructions address adaptive leakage involving bounded groups
of colluding parties~\cite{KMS19,ChattopadhyayEtAl20}, or improve share
sizes~\cite{CKOS22}. Such constructions need not preserve the algebraic
structure of the underlying scheme, including linearity, which is
important in applications such as MPC. In contrast, our work establishes
leakage resilience of Shamir scheme.

\paragraph{Local leakage resilience of existing linear schemes.}
Shamir and additive secret sharing are the most common linear schemes widely used in applications.
For additive \(n\)-out-of-\(n\) sharing over prime fields,
Maji et al.~\cite{MNPSWYY22} give a tight analysis of the
parity-of-parity attack, establishing an optimal lower bound on
the number of shares needed for resilience to one-bit leakage.
Koga and Abe~\cite{KogaAbe25} obtain tight one-bit leakage bounds
and further bounds for larger output alphabets.

\paragraph{Attacks on local leakage resilience.}
Nielsen and Simkin present one-bit attacks
at sublinear thresholds of order \(n/\log n\)~\cite{NielsenSimkin20};
Hoffmann and Simkin develop robust attacks that can tolerate noise in a
constant fraction of the leakage outputs~\cite{HoffmannSimkin23}.
Other attacks exploit particular leakage functions or prime-field
repair procedures~\cite{Nguyen24,ConEtAl24}. Boudgoust and Simkin
give a perfectly correct attack on schemes with linear reconstruction
using \(\log_2 k+O(1)\) leaked bits per share over arbitrary finite
fields~\cite{BoudgoustSimkin26}. The shifted quadratic-residue attack
from~\cite{Kopparty26} is discussed in \cref{sec:structured-leakage}. None of these
statements establishes impossibility at threshold rate half
for exactly one-bit leakage. Our logarithmic-bit guarantee against
arbitrary maps is above half; the below-half
result concerns one bit and does not cover the full sublinear-threshold range.

\paragraph{Physical-bit, Hamming-weight, and noisy leakage.}
Physical-bit leakage reveals selected bits in a fixed binary
representation of each share, rather than the output of an arbitrary
bounded-output function. Over prime fields, Maji et al.~\cite{MNPSW21}
prove resilience to multiple physical bits per share for every
reconstruction threshold \(k\ge2\), using random evaluation points
and sufficiently large fields. Subsequent work studies physical-bit
leakage over extension fields~\cite{MNPY24,Nguyen25}.
Hwang et al.~\cite{HwangEtAl25} characterize secure evaluation points
for \(n\)-out-of-\(n\) Shamir sharing over prime fields, while
Venkitesh~\cite{Venkitesh26} provides partial derandomization in certain
extension-field regimes. For public-coefficient inner-product masking,
Faust et al.~\cite{FaustEtAl25} establish physical-bit leakage guarantees
over Mersenne prime fields.

Hamming-weight leakage reveals the number of ones in a share's binary
representation. It has been analyzed for inner-product masking over
binary extension fields~\cite{BHMY26} and for additive sharing over
cyclic groups whose moduli are close to powers of
two~\cite{BHMSY26Additive,BHMSY26Energy}. Noisy-leakage models instead
constrain the leakage channel through a noise assumption rather than
its output length~\cite{ProuffRivain13,DucDziembowskiFaust14,FaustEtAl24,JahandidehEtAl25,GuptaMahdavifar26}.
These results exploit the prescribed leakage function, the
representation of shares, or the noise; they do not give resilience
to arbitrary local leakage.

% \paragraph{Efron--Stein decomposition.}
% The Efron--Stein decomposition is a standard orthogonal decomposition
% on product probability spaces~\cite{ODonnell14}. It separates the
% dependence of a function according to subsets of coordinates and is
% preserved by conditional averaging on each coordinate. Related uses
% of the decomposition occur in secure non-interactive
% simulation~\cite{KMN22}. Here we average over inputs giving the same
% leakage output. The code's Fourier expansion identifies the components
% with dual words of a prescribed support, and our transfer bounds
% those components by a dual-code count.

\paragraph{List decoding of Reed--Solomon codes.}
Sudan's interpolation algorithm and
its extension by Guruswami and Sudan established the classical
algorithmic approach and the Johnson-radius benchmark
\cite{Sudan97,GuruswamiSudan99}. Progress beyond that benchmark includes
random puncturing, generalized Singleton bounds, and improved
large-error-radius and list-recovery guarantees
\cite{RudraWootters14,ShangguanTamo20,FerberKwanSauermann22,GuoEtAl24}.
For a code of rate \(\mathsf R\), the large-alphabet capacity target
is an error fraction \(1-\mathsf R-\zeta\), for any fixed gap
\(\zeta>0\). Brakensiek, Gopi, and Makam proved that generic
Reed--Solomon codes reach this target over sufficiently large
fields~\cite{BGM23}; \cite{AGGLZ25} reduced the field-size requirement
for random evaluation points to a sufficiently large constant times
the block length. These guarantees concern generic or random codes,
not every evaluation tuple.

The recent results of~\cite{BCPZZ26,Jeronimo26} reach capacity for
arbitrarily fixed evaluation points over prime fields.
The former supplies a new interpolation-based algorithmic approach;
Jeronimo's theorem gives the field-independent list bound
\(n^{c_{\zeta,\ell}}\) for every fixed input-list size \(\ell\),
every prime \(p\ge n\), and sufficiently large \(n\).
Dao, Kominers, and Thaler~\cite[Theorem~G.2]{DKT26} also obtain
list-recovery guarantees for fixed evaluation points, providing
another route to the field-independent bound used here.
We use \(\ell=1\) for ordinary list decoding and \(\ell=2^m-1\)
for \(m\)-bit leakage. For growing \(m\), we keep the dependence
on \(\ell\) explicit in Jeronimo's proof and in DKT's finite
estimates (\cref{app:growing-lists,app:dkt-lists}). The field-independent bound matters
because our security bound must not deteriorate when \(p\) grows. Coordinate scaling
extends these guarantees to generalized Reed--Solomon dual codes.

Random linear codes provide a different route to capacity for list
decoding~\cite{GHSZ02,GHK11,LiWootters21} and list
recovery~\cite{LiShagrithaya25,DoronEtAl25,GoyalGuruswami26}; folded
Reed--Solomon codes also achieve list-recovery
capacity~\cite{GuruswamiRudra08}. For fixed alphabets, list-decoding
capacity is expressed using the appropriate entropy function and
should not be identified with the large-alphabet limit \(1-\mathsf R\).
Our transfer applies coding guarantees to the \emph{dual} code:
the same statement uses ordinary list decoding for one-bit outputs
and list recovery with input lists of size \(2^m-1\) for \(m\)-bit outputs.
This is not direct recovery of shares from the leakage. A preimage of
a bit can contain \(\Theta(p)\) symbols; the lists in
the transfer instead arise by replacing the Fourier weights
\(\|\T_i\chi_a\|_2^2\) by list indicators in a weighted counting
bound.

\paragraph{Repairing Reed--Solomon codes.}
Repair bandwidth measures the amount of data downloaded to recover
an erased symbol and is a central parameter in distributed
storage~\cite{DimakisEtAl10}. For Reed--Solomon codes, Guruswami and
Wootters~\cite{GuruswamiWootters16} develop bandwidth-efficient linear
repair schemes, and Tamo, Ye, and Barg~\cite{TamoYeBarg17} achieve the
cut-set bound over suitable large extension fields. Over prime fields,
Con, Shutty, Tamo, and Wootters~\cite{ConEtAl24} use exponential sums
to construct nonlinear repair procedures for low-rate Reed--Solomon
codes that download only a constant number of bits from each helper.
These results connect repair to local leakage: viewing the secret
as an erased evaluation, a repair procedure yields a reconstruction
attack whenever each helper's message fits the permitted leakage budget.

\paragraph{Discrepancy, zero-error list recovery, and OPI.}
The survey of Resch and Venkitesh~\cite{ReschVenkitesh25} discusses
the relation between local leakage and zero-error list recovery
above capacity. Doron et al.~\cite{DoronEtAl26} prove near-expected
codeword counts for random linear codes and derive local leakage
resilience for ramp sharing with arbitrary positive linear
reconstruction thresholds and a sublinear threshold gap.
Their polynomial field sizes are smaller than the block length,
and their codes and privacy/reconstruction tradeoff differ from
the GRS threshold setting here. This prior connection is distinct
from the dual-code transfer in \cref{sec:list-transfer}.

Sun and Wootters~\cite{SunWootters26} connect Fourier estimates for
local leakage to worst-case optimal polynomial intersection (OPI) and obtain improvements over the
semicircle satisfaction guarantee. Horinaga and Yamakawa~\cite{HorinagaYamakawa26}
improve both existential and quantum-algorithmic bounds; Jo~\cite{Jo26OPI}
gives a quantum sampler based on complete list decoding.

\paragraph{Organization.}
\Cref{sec:preliminaries} gives the preliminaries and the Shamir reduction.
\Cref{sec:decomposition} develops the orthogonal decomposition of the
leakage distance. \Cref{sec:list-transfer} proves the list-recovery
transfer and logarithmic-bit leakage resilience above half.
\Cref{sec:structured-leakage} uses the product bound for Gowers norms
to establish the restricted-leakage results at every constant threshold
rate. \Cref{sec:main-result} combines the two tools to go below half.
\Cref{app:growing-lists} gives the quantitative list-recovery bound
used for logarithmic-bit leakage; \cref{app:dkt-lists} gives an
alternative derivation from DKT.
% \Cref{app:context} gives a more detailed comparison
% with prior results and their leakage models.

\section{Preliminaries}
\label{sec:preliminaries}
\label{sec:model}

Throughout, \(p\) is an odd prime and \(\Fp\) the field with \(p\)
elements; the main theorems take \(p>n\). Logarithms are natural unless
a base is shown. Vectors are written in bold. Write
\([n]=\{1,\ldots,n\}\), \(\supp(\vec y)=\{i:y_i\ne0\}\),
\(\wt(\vec y)=|\supp(\vec y)|\), and
\(\distH(\vec x,\vec y)=\wt(\vec x-\vec y)\). The binary entropy is
\(h(t)=-t\log t-(1-t)\log(1-t)\), with \(h(0)=h(1)=0\). For a finite
set \(S\), \(U_S\) denotes the uniform distribution on \(S\); a
uniform random vector is written \(\vec U_S\).

\begin{definition}[Statistical distance]
For distributions \(P,Q\) on a finite set \(\Omega\), the statistical distance between them is defined as 
\[\SD(P,Q)=\frac12\sum_{x\in\Omega}|P(x)-Q(x)|.\]
\end{definition}

\subsection{Generalized Reed--Solomon codes}
\label{sec:code-view}

\begin{definition}[Linear codes and Reed--Solomon codes]
A linear code \(C\) of length \(n\) over \(\Fp\) is a subspace of
\(\Fp^n\); its elements are codewords. Its parameters
\([n,k,d]_p\) are its length, dimension, and minimum Hamming
distance between distinct codewords. Its rate is \(k/n\). Its dual is
\(C^\perp=\{\vec y\in\Fp^n:\sum_ix_iy_i=0\text{ for every }\vec x\in C\}\).

For distinct evaluation points \(\vec\alpha\in\Fp^n\), \(1\le k\le n\),
and multipliers \(\vec v\in(\Fp^*)^n\), the generalized Reed--Solomon
code of length \(n\) and dimension \(k\) is
\[
  \GRS_{\vec\alpha,\vec v}(n,k)
  =\{(v_iQ(\alpha_i))_{i=1}^n:\deg Q<k\}.
\]
Taking \(\vec v=\vec1\) gives an Reed--Solomon code. A code of
dimension \(k\) is maximum-distance-separable (MDS) if its minimum
distance is \(n-k+1\), or equivalently if projection onto any \(k\)
coordinates is a bijection onto \(\Fp^k\).
\end{definition}

GRS codes are MDS, and their duals are GRS codes on the same evaluation
points: for \(C=\GRS_{\vec\alpha,\vec v}(n,k-1)\) with \(2\le k\le n\),
\(C^\perp\) is a GRS code of dimension \(n-k+1\) and minimum distance
\(k\). Rescaling coordinate \(i\) by \(v_i^{-1}\) yields an ordinary
Reed--Solomon code of the same dimension and
preserves Hamming distances.

\subsection{Shamir secret sharing}
\label{sec:shamir}

\begin{definition}[Shamir secret sharing]
\label{def:shamir-sharing}
Let \(\vec\alpha\in(\Fp^*)^n\) have distinct coordinates. The
\((n,k,\vec\alpha)\)-Shamir scheme with \(n\) parties and
reconstruction threshold \(k\) is defined as follows.
\begin{enumerate}
\item \emph{Sharing:} to share \(s\in\Fp\), \(\Share_{\vec\alpha}(s)\)
  samples \(Z_1,\ldots,Z_{k-1}\) independent and uniform in \(\Fp\),
  forms \(P(X)=s+\sum_{j=1}^{k-1}Z_jX^j\), and gives party \(i\) the
  share \(P(\alpha_i)\).
\item \emph{Reconstruction:} any \(k\) shares determine \(s\) by
  Lagrange interpolation.
\end{enumerate}
\end{definition}

The shares of the zero secret are uniform on
\(C_0=\GRS_{\vec\alpha,\vec\alpha}(n,k-1)\), and the shares of secret
\(s\) are uniform on \(s\vec1+C_0\). For every secret, any \(k-1\)
shares are independent uniform field elements, so they reveal nothing
about the secret.

\subsection{Local leakage resilience}
\label{sec:leakage-model}

\begin{definition}[\(m\)-bit local leakage]
For a positive integer \(m\), an \(m\)-bit local leakage function is a
tuple \(\vec\tau=(\tau_1,\ldots,\tau_n)\) of maps
\(\tau_i\colon\Fp\to\{0,1\}^m\). For \(\vec x\in\Fp^n\) write
\(\vec\tau(\vec x)=(\tau_1(x_1),\ldots,\tau_n(x_n))\).
\end{definition}

\begin{definition}[Leakage resilience of a code]
\label{def:code-leakage-resilience}
A linear code \(C\subseteq\Fp^n\) is \(\eps\)-resilient to \(m\)-bit
local leakage if
\(\SD\bigl(\vec\tau(\vec U_C),\vec\tau(\vec U_{\Fp^n})\bigr)\le\eps\)
for every \(m\)-bit local leakage function \(\vec\tau\). We call this
statistical distance the leakage distance.
\end{definition}

Thus local leakage cannot distinguish a uniform codeword from a vector
of independent uniform field elements.

\begin{definition}[Leakage resilience of Shamir sharing]
\label{def:leakage-resilience}
The \((n,k,\vec\alpha)\)-Shamir scheme is \(\eps\)-resilient to
\(m\)-bit local leakage if
\(\SD\bigl(\vec\tau(\Share_{\vec\alpha}(s)),
           \vec\tau(\Share_{\vec\alpha}(s'))\bigr)\le\eps\)
for every \(m\)-bit local leakage function \(\vec\tau\) and all
secrets \(s,s'\in\Fp\).
\end{definition}

When no error bound is stated, resilience means that the leakage
distance is negligible in \(n\), namely smaller than every inverse
polynomial for sufficiently large \(n\).

\begin{remark}[Code monotonicity]
\label{lem:code-monotonicity}
Translations can be absorbed into the leakage maps, and statistical
distance is convex in its first argument~\cite{BDIR18,BDIR21}. Hence if
\(C\subseteq\Fp^n\) is \(\eps\)-resilient to \(m\)-bit local leakage,
the same bound against \(\vec U_{\Fp^n}\) holds for every mixture of
uniform distributions on translates of \(C\); in particular, every
linear code containing \(C\) is \(\eps\)-resilient.
\end{remark}

\begin{remark}[Threshold monotonicity]
\label{lem:threshold-monotonicity}
Increasing the reconstruction threshold preserves leakage
resilience~\cite{BDIR18,BDIR21}: if the \((n,k,\vec\alpha)\)-Shamir
scheme is \(\eps\)-resilient to \(m\)-bit local leakage, so is the
\((n,k',\vec\alpha)\)-scheme for every \(k\le k'\le n\).
\end{remark}

\begin{lemma}[Equivalence of code and Shamir leakage resilience]
\label{lem:shamir-from-code}
\label{lem:code-from-shamir}
Let \(\vec\alpha\in(\Fp^*)^n\) have distinct coordinates, let
\(\vec v\in(\Fp^*)^n\), and let \(2\le k\le n\). Resilience to
\(m\)-bit local leakage of \(\GRS_{\vec\alpha,\vec v}(n,k-1)\) and of
the \((n,k,\vec\alpha)\)-Shamir scheme are equivalent up to a linear
factor in the error. More precisely, code error \(\eps\) gives Shamir
error at most \(2\eps\), and Shamir error \(\eps\) gives code error
at most \((n-k+1)\eps\).
\end{lemma}

\begin{proof}
Code to Shamir is the reduction of~\cite{BDIR18,BDIR21}. Coordinate
scalings by nonzero elements can be absorbed into the leakage maps, so
we may take \(C_0=\GRS_{\vec\alpha,\vec\alpha}(n,k-1)\). If \(C_0\) is
\(\eps\)-resilient, applying its guarantee to the maps
\(x\mapsto\tau_i(x+s)\) gives
\(\SD\bigl(\vec\tau(\Share_{\vec\alpha}(s)),\vec\tau(\vec U_{\Fp^n})\bigr)\le\eps\)
for every secret \(s\), and the triangle inequality gives Shamir error
at most \(2\eps\). For the zero-secret code \(C_0\), this direction
holds for any leakage class closed under coordinatewise translations.
For arbitrary multipliers, closure under nonzero coordinate scalings
is also needed.

Conversely, suppose the Shamir scheme is \(\eps\)-resilient at
threshold \(k\), hence at every larger threshold by
\cref{lem:threshold-monotonicity}. For \(k-1\le r<n\) and multipliers
\(\vec v\), put \(C=\GRS_{\vec\alpha,\vec v}(n,r)\) and
\(C'=\GRS_{\vec\alpha,\vec w}(n,r+1)\) with \(w_i=v_i/\alpha_i\). At
threshold \(r+1\), multiplying a sharing of zero coordinatewise by
\(\vec w\) gives \(\vec U_C\), and doing the same to a sharing of a
uniform secret gives \(\vec U_{C'}\), since all \(r+1\) coefficients
are then independent and uniform. Hence, with
\(\sigma_i(x)=\tau_i(w_ix)\) and sharings at threshold \(r+1\),
convexity gives
\[
  \SD\bigl(\vec\tau(\vec U_C),\vec\tau(\vec U_{C'})\bigr)
  \le\frac1p\sum_{s\in\Fp}
     \SD\bigl(\vec\sigma(\Share_{\vec\alpha}(0)),
              \vec\sigma(\Share_{\vec\alpha}(s))\bigr)
  \le\eps .
\]
Starting with the given multipliers \(\vec v\), these steps form the
nested chain
\[
\begin{aligned}
  \GRS_{\vec\alpha,\vec v}(n,k-1)
  &\subseteq\GRS_{\vec\alpha,\vec v/\vec\alpha}(n,k)
  \subseteq\cdots\\
  &\subseteq
  \GRS_{\vec\alpha,\vec v/\vec\alpha^{\,n-k+1}}(n,n)
  =\Fp^n,
\end{aligned}
\]
where division and powers of the multipliers are coordinatewise.
The triangle inequality over the \(n-k+1\) steps gives code error
at most \((n-k+1)\eps\).
\end{proof}

\subsection{Fourier analysis and Gowers norms}
\label{sec:fourier-gowers}

For \(f\colon\Fp\to\C\), write \(\norm f_q=(\E_x|f(x)|^q)^{1/q}\) for
\(1\le q<\infty\) and \(\norm f_\infty=\max_x|f(x)|\).

\begin{definition}[Fourier transform]
Let \(e_p(t)=\exp(2\pi it/p)\) and \(\chi_a(x)=e_p(ax)\) for
\(a\in\Fp\); the index \(a\) is a frequency. The normalized Fourier
transform is \(\widehat f(a)=\E_x[f(x)\overline{\chi_a(x)}]\). For
vectors, \(\chi_{\vec y}(\vec x)=e_p(\sum_iy_ix_i)\) and
\(\widehat f(\vec y)=\E_{\vec x\in\Fp^n}
 [f(\vec x)\overline{\chi_{\vec y}(\vec x)}]\).
\end{definition}

Orthogonality of the characters gives Fourier inversion and
Parseval's identity:
\[
  f(x)=\sum_{a\in\Fp}\widehat f(a)\chi_a(x),
  \qquad
  \norm f_2^2=\sum_{a\in\Fp}|\widehat f(a)|^2 .
\]

\paragraph{Averages over a linear code.}
For a linear code \(C\subseteq\Fp^n\), a set \(A\subseteq[n]\), and
functions \(\vec f_A=(f_i\colon\Fp\to\C)_{i\in A}\), write
\[
  \Gamma_C(\vec f_A)
  =\E_{\vec X\sim U_C}\Bigl[\prod_{i\in A}f_i(X_i)\Bigr],
\]
with \(\Gamma_C(\vec f_\emptyset)=1\). We use the Poisson summation
formula for a linear code~\cite[Lemma~4.16]{BDIR21}.

\begin{lemma}[Dual-code formula]
\label{lem:dual-code-formula}
Let \(C\subseteq\Fp^n\) be a linear code. For every
\(\vec x\in\Fp^n\),
\[
  \frac{p^n}{|C|}\one[\vec x\in C]
  =\sum_{\vec y\in C^\perp}\chi_{\vec y}(\vec x).
\]
\end{lemma}

The identity follows from Fourier inversion and
\(\E_{\vec X\sim U_C}[\chi_{\vec y}(\vec X)]=\one[\vec y\in C^\perp]\).

\begin{definition}[Gowers norms]
\label{def:gowers-norms}
For a positive integer \(t\) and \(f\colon\Fp\to\C\),
\[
  \norm f_{U^t}^{2^t}
  =\E_{x,h_1,\ldots,h_t\in\Fp}\Bigl[
    \prod_{\vec\omega\in\{0,1\}^t}
      \comp^{|\vec\omega|}f(x+\vec\omega\cdot\vec h)\Bigr],
\]
where \(\vec h=(h_1,\ldots,h_t)\), \(\comp\) denotes complex
conjugation, \(|\vec\omega|=\omega_1+\cdots+\omega_t\), and the
\(t+1\) averaging variables are independent and uniform. The same
definition applies to functions on any finite abelian group, in
particular on \(\Fp^J\) for a finite set \(J\).
\end{definition}

The conjugations make these averages real and nonnegative; the
\(U^2\) and \(U^3\) norms average products of four and eight
translated values. For \(t\ge2\), \(\norm\cdot_{U^t}\) is a norm, and
\(|\E f|=\norm f_{U^1}\le\norm f_{U^2}\le\cdots\le\norm f_\infty\)~\cite{TaoVu06}.
Fourier orthogonality gives
\(\norm f_{U^2}^4=\sum_{a\in\Fp}|\widehat f(a)|^4\), and with
\(\Delta_hf(x)=f(x+h)\overline{f(x)}\) the definition gives
\(\norm f_{U^3}^8=\E_h\bigl[\norm{\Delta_hf}_{U^2}^4\bigr]\).

\subsection{List decoding and list recovery}
\label{sec:list-decoding}

\begin{definition}[List size]
\label{def:list-size}
For a code \(D\subseteq\Fp^n\), an integer radius \(0\le R\le n\), and
an input-list size \(\ell\ge1\),
\[
  \List_D(R,\ell)
  =\max_{\substack{S_1,\ldots,S_n\subseteq\Fp\\|S_i|\le\ell}}
    \bigl|\{\vec y\in D:|\{i:y_i\in S_i\}|\ge n-R\}\bigr| .
\]
Write \(\List_D(R)=\List_D(R,1)\).
\end{definition}

Thus \(\List_D(R,\ell)\le L\) means that every choice of input lists
admits at most \(L\) codewords agreeing with them in at least \(n-R\)
coordinates (list recovery). For \(\ell=1\),
\(\List_D(R)=\max_{\vec y\in\Fp^n}|\{\vec v\in D:\distH(\vec v,\vec y)\le R\}|\)
(list decoding). We use only these counting bounds, not decoding
algorithms.

We use a field-size-independent list-recovery bound supplied by
Jeronimo~\cite[Theorem~1.1]{Jeronimo26}, building on the decoding
method of Brakensiek et al.~\cite{BCPZZ26}, and also by Dao, Kominers,
and Thaler~\cite[Theorem~G.2]{DKT26}. For growing leakage alphabets,
we need the dependence on the input-list size explicit.
The following quantitative form is derived from Jeronimo's proof in
\cref{app:growing-lists} and from DKT's finite estimates in
\cref{app:dkt-lists}.

\begin{theorem}[Reed--Solomon list recovery, derived from~\cite{Jeronimo26,DKT26}]
\label{thm:rs-list-decoding}
Fix \(0<\zetaGap<1\). There is a constant \(c_{\zetaGap}\ge1\) such
that, for every integer \(\ell\ge1\), every
\(n\ge c_{\zetaGap}\ell^{9/\zetaGap}\), every prime \(p>n\), and every integer
\(1\le h\le(1-\zetaGap)n\),
\[
  \List_{\GRS_{\vec\alpha,\vec v}(n,h)}
    \bigl(\floor{n-h-\zetaGap n},\ell\bigr)
  \le n^{c_{\zetaGap}\ell^{3/\zetaGap}} .
\]
\end{theorem}

The exponent is independent of \(p\), so there is no upper bound on
the field size; the gap \(\zetaGap\) is fixed, while \(\ell\) may grow
subject to the length condition. The statement for ordinary
Reed--Solomon codes gives the GRS version by replacing each input list
\(S_i\) with \(v_i^{-1}S_i\), which preserves list sizes and agreement
counts.

For \(C=\GRS_{\vec\alpha,\vec v}(n,k-1)\) with \(2\le k\le n\), the
dual has dimension \(n-k+1\), which is at most \((1-\zetaGap)n\) when
\(k-1\ge\zetaGap n\), and the radius in the theorem is then exactly
\(\floor{k-1-\zetaGap n}\). Hence, whenever \(k-1\ge\zetaGap n\) and
\(n\ge c_{\zetaGap}\ell^{9/\zetaGap}\),
\[
  \List_{C^\perp}\bigl(\floor{k-1-\zetaGap n},\ell\bigr)
  \le n^{c_{\zetaGap}\ell^{3/\zetaGap}} .
\]
Taking \(\ell=1\) gives the corresponding list-decoding bound.
All later estimates are uniform in the leakage maps, so no union bound
over maps is needed.

For comparison, \cite[Corollary~5.1]{BCPZZ26} gives list size
\(p^{O(1)}\) at fixed rate and gap when \(p\ge Cn\), which also
suffices for our arguments when \(\log p=o(n)\);
\cref{thm:rs-list-decoding} removes this upper restriction and allows
every prime \(p>n\).

\section{An orthogonal decomposition of the leakage distance}
\label{sec:decomposition}

This section writes the leakage distance as a norm of one function of
the outputs and decomposes that function into orthogonal components
indexed by sets of coordinates; \cref{sec:expansion} specializes to
one-bit leakage.

\subsection{Efron--Stein decomposition}
\label{sec:efron-stein}

Throughout this subsection, \(\Omega=\prod_{i=1}^n\Omega_i\) is a product
of finite sets and \(\mu=\prod_{i=1}^n\mu_i\) is a product probability
measure on \(\Omega\) in which every \(\mu_i\) has full support. For
\(f,g\colon\Omega\to\C\) write
\(\langle f,g\rangle_\mu=\E_\mu[f\overline g]\) and
\(\|f\|_{q,\mu}=(\E_\mu|f|^q)^{1/q}\) for \(1\le q<\infty\); we omit the
measure when it is clear from context.

\begin{definition}[Efron--Stein decomposition {\cite[Chapter~8]{ODonnell14}}]
\label{def:efron-stein}
Every \(f\colon\Omega\to\C\) has a unique expansion
\[
  f=\sum_{A\subseteq[n]}f^{=A}
\]
in which \(f^{=A}\) depends only on the coordinates in \(A\) and
\(\E_{x_i\sim\mu_i}[f^{=A}(\vec x)]=0\) for every \(i\in A\) and every
fixing of the other coordinates. The components are pairwise orthogonal
in \(L^2(\Omega,\mu)\); in particular
\[
  f^{=\emptyset}=\E_\mu[f]
  \qquad\text{and}\qquad
  \|f\|_{2,\mu}^2=\sum_{A\subseteq[n]}\|f^{=A}\|_{2,\mu}^2 .
\]
\end{definition}

Concretely, choose for each \(i\) an orthonormal basis of
\(L^2(\Omega_i,\mu_i)\) that contains the constant function \(1\). Their
products, one factor per coordinate, form an orthonormal basis of
\(L^2(\Omega,\mu)\), and \(f^{=A}\) is the part of the expansion of \(f\)
whose factors are nonconstant at exactly the coordinates in
\(A\)~\cite[Chapter~8]{ODonnell14}. Equivalently, let \(E_i^\mu\)
average coordinate \(i\) according to \(\mu_i\), leaving the others
fixed. Then
\[
  f^{=A}
  =\Bigl(\prod_{i\in A}(\mathrm{Id}-E_i^\mu)\Bigr)
   \Bigl(\prod_{i\notin A}E_i^\mu\Bigr)f.
\]
On \(\Fp^n\) with the uniform
measure and the character basis this reads as
follows~\cite[Proposition~8.36]{ODonnell14}; see
also~\cite[Proposition~7]{KMN22}.

\begin{lemma}[Fourier form of the components]
\label{lem:fourier-components}
For every \(f\colon\Fp^n\to\C\) and \(A\subseteq[n]\), under the uniform
measure on \(\Fp^n\),
\[
  f^{=A}
  =\sum_{\vec y\in\Fp^n\colon\supp(\vec y)=A}
   \widehat f(\vec y)\,\chi_{\vec y}.
\]
\end{lemma}

Leakage enters through conditional expectation.

\begin{definition}[Markov operator~{\cite[Section~4.4]{KMN22}}]
\label{def:leakage-operator}
Let \((X,Y)\) be a pair of random variables on finite sets
\(\Omega\times\Omega'\) with marginal distributions \(\mu\) and \(\nu\). Its
\emph{Markov operator} \(\T\colon L^2(\Omega,\mu)\to L^2(\Omega',\nu)\) is
\[
  (\T f)(y)=\E[f(X)\mid Y=y],\qquad y\in\supp(\nu).
\]
\end{definition}

The operator \(\T\) is linear, \(\T1=1\), and \(\E_\nu[\T f]=\E_\mu[f]\).
If \((X_i,Y_i)_{i\in[n]}\) are independent pairs with Markov operators
\(\T_i\), then the Markov operator of \(\vec X=(X_i)_i\) and
\(\vec Y=(Y_i)_i\) is the tensor product \(\T=\bigotimes_i\T_i\): for
\(f(\vec x)=\prod_i f_i(x_i)\) one has
\((\T f)(\vec y)=\prod_i(\T_i f_i)(y_i)\). We use two standard
properties~\cite[Proposition~5]{KMN22}.

\begin{lemma}[Markov operators and components]
\label{lem:product-expectation}
Let \(\T\) be the Markov operator of a pair \((X,Y)\) with marginal distributions
\(\mu,\nu\).
\begin{enumerate}
\item \(\|\T f\|_{2,\nu}\le\|f\|_{2,\mu}\) for every
  \(f\in L^2(\Omega,\mu)\).
\item If \((X,Y)\) is a product of independent pairs as above, then
  \((\T f)^{=A}=\T(f^{=A})\) for every \(f\) and every \(A\subseteq[n]\),
  where the components of \(f\) are taken under \(\mu\) and those of
  \(\T f\) under \(\nu\).
\end{enumerate}
\end{lemma}

\begin{proof}
Conditional Jensen gives
\[
  \|\T f\|_{2,\nu}^2
  =\E_Y\bigl|\E[f(X)\mid Y]\bigr|^2
  \le\E_Y\bigl[\E[|f(X)|^2\mid Y]\bigr]
  =\|f\|_{2,\mu}^2.
\]
For independent coordinate pairs, averaging coordinate \(i\) before
or after applying its Markov operator gives the same constant:
\(\T E_i^\mu=E_i^\nu\T\). Applying this identity to the product
formula for the components above gives
\[
  \T(f^{=A})
  =\Bigl(\prod_{i\in A}(\mathrm{Id}-E_i^\nu)\Bigr)
   \Bigl(\prod_{i\notin A}E_i^\nu\Bigr)\T f
  =(\T f)^{=A}. \qedhere
\]
\end{proof}

\subsection{The leakage distance as a norm on the output space}
\label{sec:transfer-expansion}

Fix a linear code \(C\subseteq\Fp^n\) and local maps
\(\tau_i\colon\Fp\to\mathcal B_i\) into finite sets; \(m\)-bit leakage
is the case \(\mathcal B_i=\{0,1\}^m\). Replacing each \(\mathcal B_i\)
by the image \(\tau_i(\Fp)\), we may assume that every output symbol
occurs, so that
\[
  \nu_i(b)=\Prb[\tau_i(U_{\Fp})=b]>0
  \quad\text{for all } b\in\mathcal B_i,
  \qquad
  \nu(\vec b)=\prod_{i=1}^n\nu_i(b_i).
\]
Thus \(\nu\) is the reference leakage distribution, namely the distribution
of \(\vec\tau(\vec U_{\Fp^n})\). Let \(\T_i\) be the Markov operator of
\cref{def:leakage-operator} for the pair
\(\bigl(U_{\Fp},\tau_i(U_{\Fp})\bigr)\), and \(\T\) the one for
\(\bigl(\vec U_{\Fp^n},\vec\tau(\vec U_{\Fp^n})\bigr)\); the pairs are
independent, so \(\T=\bigotimes_i\T_i\). Norms of functions on the
inputs are taken under the uniform measure, and norms of functions on
the outputs under \(\nu\), or under \(\nu_i\) for a function of one
coordinate.

The object we decompose is the \emph{normalized indicator}
\(f_C\colon\Fp^n\to\mathbb R\),
\[
  f_C(\vec x)
  =\frac{\Prb[\vec U_C=\vec x]}{\Prb[\vec U_{\Fp^n}=\vec x]}
  =\frac{p^n}{|C|}\,\one[\vec x\in C],
\]
the likelihood ratio of the two input distributions. Its conditional
mean given the leakage is the likelihood ratio of the two leakage
distributions.

\begin{lemma}[Leakage probabilities and statistical distance]
\label{lem:leakage-distance}
For every \(\vec b\in\prod_i\mathcal B_i\),
\[
  \Prb[\vec\tau(\vec U_C)=\vec b]=\nu(\vec b)\,(\T f_C)(\vec b),
  \qquad
  \Prb[\vec\tau(\vec U_{\Fp^n})=\vec b]=\nu(\vec b).
\]
Consequently,
\[
  \SD\bigl(\vec\tau(\vec U_C),\vec\tau(\vec U_{\Fp^n})\bigr)
  =\frac12\,\|\T f_C-1\|_{1,\nu}.
\]
\end{lemma}

\begin{proof}
Write \(\vec U=\vec U_{\Fp^n}\). Since
\(\Prb[\vec U_C=\vec x]=p^{-n}f_C(\vec x)\),
\[
  \Prb[\vec\tau(\vec U_C)=\vec b]
  =\sum_{\vec x\colon\vec\tau(\vec x)=\vec b}p^{-n}f_C(\vec x)
  =\Prb[\vec\tau(\vec U)=\vec b]\;
   \E\bigl[f_C(\vec U)\bigm|\vec\tau(\vec U)=\vec b\bigr].
\]
The second factor is \((\T f_C)(\vec b)\) by definition, and the first
is \(\nu(\vec b)\) because the coordinates of \(\vec\tau(\vec U)\) are
independent with distributions \(\nu_i\). Substituting both identities into the
definition of statistical distance gives
\(\frac12\sum_{\vec b}\nu(\vec b)\,|(\T f_C)(\vec b)-1|\).
\end{proof}

\subsection{The orthogonal components}

For the rest of this section we assume that \(C^\perp\) has minimum
distance at least \(k\). This is the only property of \(C\) used below;
for \(C=\GRS_{\vec\alpha,\vec v}(n,k-1)\) it holds because \(C^\perp\)
is a GRS code of dimension \(n-k+1\).

Since \(\T\) is a product Markov operator,
\cref{lem:product-expectation} gives \((\T f_C)^{=A}=\T(f_C^{=A})\), so
it remains to identify the input components. By
\cref{lem:dual-code-formula},
\(\widehat f_C(\vec y)=\one[\vec y\in C^\perp]\) for every
\(\vec y\in\Fp^n\), and \cref{lem:fourier-components} then gives
\begin{equation}
\label{eq:code-components}
  f_C^{=A}
  =\sum_{\vec y\in\Fp^n\colon\supp(\vec y)=A}
     \widehat f_C(\vec y)\,\chi_{\vec y}
  =\sum_{\vec y\in C^\perp\colon\supp(\vec y)=A}\chi_{\vec y}.
\end{equation}
In particular \(f_C^{=\emptyset}=1\), and \(f_C^{=A}=0\) whenever
\(1\le|A|<k\), because every nonzero dual word has weight at least
\(k\). Since \(\T1=1\),
\[
  \T f_C-1=\sum_{A\subseteq[n]\colon|A|\ge k}\T(f_C^{=A}),
\]
and the summands are pairwise orthogonal in \(L^2(\nu)\), being the
components of \(\T f_C\) under \(\nu\) (\cref{lem:product-expectation}
and \cref{def:efron-stein}).

List decoding of \(C^\perp\) controls squared \(L^2\) norms of
components, which add across supports by orthogonality
(\cref{sec:list-transfer,sec:main-result}); Gowers-norm estimates
control \(L^1\) norms one support at a time
(\cref{sec:structured-leakage,sec:main-result}). The next proposition
allows any cutoff between the two regimes, and \cref{prop:full-l2-bound}
records the extreme cutoffs.

\begin{proposition}[Splitting the leakage distance]
\label{prop:leakage-decomposition}
For every integer \(r_0\) with \(k-1\le r_0\le n\),
\[
  \SD\bigl(\vec\tau(\vec U_C),\vec\tau(\vec U_{\Fp^n})\bigr)
  \le\frac12\Bigl(\sum_{k\le|A|\le r_0}\|\T(f_C^{=A})\|_2^2\Bigr)^{1/2}
   +\frac12\sum_{|A|>r_0}\|\T(f_C^{=A})\|_1 .
\]
\end{proposition}

\begin{proof}
Split the decomposition of \(\T f_C-1\) at \(|A|=r_0\) and
apply the triangle inequality in \(L^1(\nu)\) to the two parts. For the
supports of size at most \(r_0\), Cauchy--Schwarz gives
\(\|\cdot\|_1\le\|\cdot\|_2\), and by orthogonality the squared
\(L^2\) norm of the sum is the sum of the squared norms. For the larger
supports, the triangle inequality bounds the \(L^1\) norm of the sum by
the sum of the \(L^1\) norms. \Cref{lem:leakage-distance} supplies the
factor \(\frac12\).
\end{proof}

The extreme cutoffs \(r_0=n\) and \(r_0=k-1\) give the following.

\begin{corollary}[The two extreme cutoffs]
\label{prop:full-l2-bound}
The following bounds hold.
\[
  \SD\bigl(\vec\tau(\vec U_C),\vec\tau(\vec U_{\Fp^n})\bigr)
  \le\frac12\Bigl(\sum_{|A|\ge k}\|\T(f_C^{=A})\|_2^2\Bigr)^{1/2}
  \ \ \text{and}\ \
  \SD\bigl(\vec\tau(\vec U_C),\vec\tau(\vec U_{\Fp^n})\bigr)
  \le\frac12\sum_{|A|\ge k}\|\T(f_C^{=A})\|_1 .
\]
\end{corollary}

The large-support bound has an exact form for arbitrary output
alphabets in terms of centered output indicators.

\begin{lemma}[Output-indicator form of the components]
\label{lem:output-indicator-components}
For \(i\in[n]\) and \(b\in\mathcal B_i\), let
\[
  f_{i,b}(x)=\one[\tau_i(x)=b]-\nu_i(b).
\]
These functions have mean zero and \(\|f_{i,b}\|_\infty\le1\). For
every \(A\subseteq[n]\),
\[
  \|\T(f_C^{=A})\|_{1,\nu}
  =\sum_{\vec b_A\in\prod_{i\in A}\mathcal B_i}
    \left|\E_{\vec X\sim U_C}
      \left[\prod_{i\in A}f_{i,b_i}(X_i)\right]\right|.
\]
\end{lemma}

\begin{proof}
Expanding the product formula in \cref{sec:efron-stein} gives
\[
  (\T f_C)^{=A}
  =\sum_{S\subseteq A}(-1)^{|A\setminus S|}
     \Bigl(\prod_{i\notin S}E_i^\nu\Bigr)\T f_C.
\]
Each term averages all output coordinates outside \(S\). By
\cref{lem:leakage-distance}, each marginal of \(\T f_C\) is the
corresponding marginal leakage probability divided by its reference
probability. Hence
\[
\begin{aligned}
  (\T f_C^{=A})(\vec b_A)
  &=\sum_{S\subseteq A}(-1)^{|A\setminus S|}
    \frac{\Prb[\vec\tau_S(\vec U_C)=\vec b_S]}
         {\prod_{i\in S}\nu_i(b_i)}\\
  &=\E_{\vec X\sim U_C}
    \left[\prod_{i\in A}
      \left(\frac{\one[\tau_i(X_i)=b_i]}{\nu_i(b_i)}-1\right)\right]
   =\frac{\E_{\vec X\sim U_C}
          [\prod_{i\in A}f_{i,b_i}(X_i)]}
          {\prod_{i\in A}\nu_i(b_i)}.
\end{aligned}
\]
Taking the \(L^1(\nu)\) norm cancels the denominator and proves the
identity. No lower bound on the positive probabilities \(\nu_i(b)\)
is needed; unused output symbols have \(f_{i,b}=0\) and may also be
included in the sum.
\end{proof}

\subsection{One-bit leakage}
\label{sec:expansion}

Now let \(\tau_i\colon\Fp\to\{0,1\}\). Each \(L^2(\nu_i)\) is then
spanned by \(1\) and one centered function, so every output component
is a scalar multiple of a single product function; the scalar is an
average over the code of centered leakage functions.

Put \(\theta_i=\E_x[\tau_i(x)]\). If \(\tau_i\) is nonconstant, so that
\(0<\theta_i<1\), define
\[
  g_i(x)=\frac{\tau_i(x)-\theta_i}{\sqrt{\theta_i(1-\theta_i)}},
  \qquad
  f_i(x)=2\bigl(\tau_i(x)-\theta_i\bigr);
\]
if \(\tau_i\) is constant, put \(g_i=f_i=0\). Both functions have mean
zero, \(\|g_i\|_2\le1\) with equality for nonconstant \(\tau_i\), and
\(\|f_i\|_\infty\le2\). Write \(\vec g_A=(g_i)_{i\in A}\) and
\(\vec f_A=(f_i)_{i\in A}\), and recall that
\(\Gamma_C(\vec h_A)=\E_{\vec X\sim U_C}\bigl[\prod_{i\in A}h_i(X_i)\bigr]\)
for functions \(h_i\colon\Fp\to\C\).

\begin{lemma}[Norms of the binary components]
\label{lem:binary-components}
For every \(A\subseteq[n]\),
\[
  \|\T(f_C^{=A})\|_2^2=|\Gamma_C(\vec g_A)|^2
  \qquad\text{and}\qquad
  \|\T(f_C^{=A})\|_1=|\Gamma_C(\vec f_A)| .
\]
Both sides vanish when \(1\le|A|<k\), and when \(A\) contains a
coordinate whose map is constant. Moreover,
\(\|\T_i\chi_a\|_2^2=|\widehat g_i(a)|^2\) for every \(i\) and every
\(a\ne0\).
\end{lemma}

\begin{proof}
If \(\tau_i\) is constant, then \(L^2(\nu_i)\) consists of the constant
functions, so every output component with \(i\in A\) vanishes, as do
\(\Gamma_C(\vec g_A)\) and \(\Gamma_C(\vec f_A)\) by the convention
\(g_i=f_i=0\); also \(\T_i\chi_a=\E_x[\chi_a(x)]=0=\widehat g_i(a)\) for
\(a\ne0\). We may therefore assume that every \(\tau_i\) is nonconstant.

For an output bit put
\(\phi_i(b)=(b-\theta_i)/\sqrt{\theta_i(1-\theta_i)}\), so that
\(g_i=\phi_i\circ\tau_i\), and let
\(\phi_A(\vec b)=\prod_{i\in A}\phi_i(b_i)\). The functions \(1,\phi_i\)
form an orthonormal basis of \(L^2(\nu_i)\), so the product basis of
\cref{sec:efron-stein} contains exactly one function nonconstant
precisely on \(A\), namely \(\phi_A\), and the component of \(\T f_C\)
on \(A\) is \(\langle\T f_C,\phi_A\rangle_\nu\,\phi_A\). By
\cref{lem:leakage-distance},
\(\nu(\vec b)(\T f_C)(\vec b)=\Prb[\vec\tau(\vec X)=\vec b]\) for
\(\vec X\sim U_C\), so
\[
  \langle\T f_C,\phi_A\rangle_\nu
  =\E_{\vec X\sim U_C}\bigl[\phi_A(\vec\tau(\vec X))\bigr]
  =\Gamma_C(\vec g_A).
\]
Together with \((\T f_C)^{=A}=\T(f_C^{=A})\) this gives
\(\T(f_C^{=A})=\Gamma_C(\vec g_A)\,\phi_A\), and
\(\|\phi_A\|_{2,\nu}=1\) proves the first identity. For the second,
independence under \(\nu\) and
\(\E_{\nu_i}|\phi_i|=2\sqrt{\theta_i(1-\theta_i)}\) give
\[
  \|\T(f_C^{=A})\|_{1,\nu}
  =|\Gamma_C(\vec g_A)|\prod_{i\in A}2\sqrt{\theta_i(1-\theta_i)}
  =|\Gamma_C(\vec f_A)|,
\]
since \(f_i=2\sqrt{\theta_i(1-\theta_i)}\,g_i\). Vanishing for
\(1\le|A|<k\) is \eqref{eq:code-components} together with the
dual-distance assumption.

Finally, for \(a\ne0\) the function \(\T_i\chi_a\) has mean zero and is
therefore a multiple of \(\phi_i\), with coefficient
\(\langle\T_i\chi_a,\phi_i\rangle_{\nu_i}
 =\E_x[\chi_a(x)g_i(x)]=\widehat g_i(-a)\).
Since \(g_i\) is real, \(|\widehat g_i(-a)|=|\widehat g_i(a)|\).
\end{proof}

Substituting the lemma into \cref{prop:leakage-decomposition} gives the
form of the decomposition used for one-bit leakage.

\begin{proposition}[Splitting small and large supports]
\label{prop:master-inequality}
For every integer \(r_0\) with \(k-1\le r_0\le n\),
\begin{equation}
\label{eq:master-inequality}
  \SD\bigl(\vec\tau(\vec U_C),\vec\tau(\vec U_{\Fp^n})\bigr)
  \le\frac12\Bigl(\sum_{k\le|A|\le r_0}|\Gamma_C(\vec g_A)|^2\Bigr)^{1/2}
   +\frac12\sum_{|A|>r_0}|\Gamma_C(\vec f_A)| .
\end{equation}
\end{proposition}

By \cref{lem:binary-components}, supports containing a constant-map
coordinate contribute nothing to \eqref{eq:master-inequality}; this is
how constant maps are handled from now on.

\section{Dual-code list recovery and resilience above half}
\label{sec:list-transfer}

This section proves that list recovery of \(C^\perp\) with input lists
of size \(Q-1\) implies resilience to leakage with \(Q\) output symbols
per coordinate; for one-bit leakage the hypothesis is ordinary list
decoding. \Cref{lem:component-bound} bounds each component of
\cref{sec:decomposition} by a weighted count of dual codewords,
\cref{lem:component-list-recovery} replaces the weights by list
indicators, and \cref{thm:grs-constant-bits} follows at the end.

\begin{theorem}[List recovery implies leakage resilience]
\label{thm:list-recovery-transfer}
For every integer \(Q\ge2\) there is a constant
\(0<\eta_Q\le1-1/(16Q^2)\), with \(\eta_2=1/2\), such that the
following holds. Let \(p\ge16Q\) be prime, let \(C\subseteq\Fp^n\) be a
linear code, and let \(R\) be an integer with \(n/2<R<n\) such that
every nonzero word of \(C^\perp\) has weight at least \(R+1\). Then for
every tuple of local maps \(\tau_i\colon\Fp\to\mathcal B_i\) with
\(|\mathcal B_i|\le Q\),
\[
  \SD\bigl(\vec\tau(\vec U_C),\vec\tau(\vec U_{\Fp^n})\bigr)
  \le\frac12\,\eta_Q^{\,R-n/2}\sqrt{\List_{C^\perp}(R,Q-1)} .
\]
For \(Q=2\) the conclusion holds for every odd prime \(p\).
\end{theorem}

The constant \(\eta_Q\) is given explicitly in
\cref{lem:local-fourier-energy}. The theorem uses only the dual
distance of \(C\); no MDS property is needed.

Throughout this section, \(C\subseteq\Fp^n\) is a linear code whose
dual has minimum distance at least \(k\), and \(f_C\), \(\T\),
\(\T_i\), \(\nu\), \(\nu_i\) are as in \cref{sec:decomposition} for
local maps \(\tau_i\colon\Fp\to\mathcal B_i\) in which every output
symbol occurs. We write \(q_i=|\mathcal B_i|\) and assume \(q_i\le Q\)
for a fixed integer \(Q\ge2\).

\subsection{Bounding the leakage components}

We apply the leakage on only \(k-1\) coordinates of a support. The dual
distance ensures that the remaining coordinates distinguish the
contributing dual words, so their characters stay orthogonal; applying
the leakage on those coordinates afterwards can only decrease the
\(L^2\) norm.

\begin{lemma}[Bounding a leakage component]
\label{lem:component-bound}
Let \(A\subseteq[n]\) have size at least \(k\), and let \(J\subseteq A\)
have size \(|A|-k+1\). Then
\[
  \|\T(f_C^{=A})\|_2^2
  \le\sum_{\vec y\in C^\perp\colon\supp(\vec y)=A}
     \prod_{i\in A\setminus J}\|\T_i\chi_{y_i}\|_2^2 .
\]
\end{lemma}

\begin{proof}
Put \(I=A\setminus J\), so \(|I|=k-1\), and write
\(\nu_I=\prod_{i\in I}\nu_i\). Since \(f_C^{=A}\) depends only on the
coordinates in \(A\), we may regard \(\T\) as \(\T_I\otimes\T_J\) with
\(\T_I=\bigotimes_{i\in I}\T_i\) and \(\T_J=\bigotimes_{i\in J}\T_i\).

First, fix \(\vec b_I\) and apply \cref{lem:product-expectation}(1) to
\(\vec x_J\mapsto(\T_I f_C^{=A})(\vec x_J,\vec b_I)\); averaging over
\(\vec b_I\sim\nu_I\) gives
\[
  \|\T(f_C^{=A})\|_2^2\le\bigl\|\T_I(f_C^{=A})\bigr\|_2^2 ,
\]
where the norm on the right is under \(\nu_I\) on the outputs of \(I\)
and the uniform measure on \(\vec x_J\in\Fp^J\).

Second, expanding \(f_C^{=A}\) by \eqref{eq:code-components} and
applying \(\T_I\) factor by factor,
\[
  (\T_I f_C^{=A})(\vec x_J,\vec b_I)
  =\sum_{\vec y\in C^\perp\colon\supp(\vec y)=A}
     \chi_{\vec y_J}(\vec x_J)\prod_{i\in I}(\T_i\chi_{y_i})(b_i).
\]
Two dual words in the sum that agree on \(J\) differ on at most
\(|I|=k-1\) coordinates, hence coincide by the dual-distance
assumption. The characters \(\chi_{\vec y_J}\) are therefore distinct,
and Parseval in \(\vec x_J\) gives
\[
  \bigl\|\T_I(f_C^{=A})\bigr\|_2^2
  =\E_{\vec b_I\sim\nu_I}
   \sum_{\vec y\in C^\perp\colon\supp(\vec y)=A}
     \prod_{i\in I}|(\T_i\chi_{y_i})(b_i)|^2
  =\sum_{\vec y\in C^\perp\colon\supp(\vec y)=A}
     \prod_{i\in I}\|\T_i\chi_{y_i}\|_2^2 ,
\]
by independence of the outputs under \(\nu_I\).
\end{proof}

The bound leaves \(k-1\) factors of the form \(\|\T_i\chi_a\|_2^2\). We
need two facts about them: their sum over nonzero frequencies is at
most \(Q-1\), which will let us replace the weights by lists of size
\(Q-1\); and each is at most \(\eta_Q<1\), which supplies the decay.

\begin{lemma}[Fourier weight after leakage]
\label{lem:local-fourier-energy}
For every \(i\) and every \(a\in\Fp\), with norms of \(\T_i\chi_a\)
taken under \(\nu_i\),
\[
  \T_i\chi_0=1,\qquad
  \|\T_i\chi_a\|_{2,\nu_i}=\|\T_i\chi_{-a}\|_{2,\nu_i}\le1,\qquad
  \sum_{a\ne0}\|\T_i\chi_a\|_{2,\nu_i}^2=q_i-1 .
\]
If \(p\ge16Q\), or if \(Q=2\) and \(p\) is odd, then for every
\(a\ne0\),
\[
  \|\T_i\chi_a\|_2^2\le\eta_Q
  :=\min\Bigl\{\frac{Q-1}{2},\;\cos^2\frac{\pi}{8Q}\Bigr\}
  \le1-\frac1{16Q^2} .
\]
In particular \(\eta_2=1/2\).
\end{lemma}

\begin{proof}
Let \(S_{i,b}=\{x\in\Fp:\tau_i(x)=b\}\), so \(\nu_i(b)=|S_{i,b}|/p\)
and \((\T_i\chi_a)(b)=\widehat{\one_{S_{i,b}}}(-a)/\nu_i(b)\). Hence
\[
  \|\T_i\chi_a\|_2^2
  =\sum_{b\in\mathcal B_i}
    \frac{|\widehat{\one_{S_{i,b}}}(a)|^2}{\nu_i(b)} ,
  \qquad
  \sum_{a\in\Fp}\|\T_i\chi_a\|_2^2
  =\sum_{b\in\mathcal B_i}\frac{\|\one_{S_{i,b}}\|_2^2}{\nu_i(b)}
  =q_i ,
\]
by Parseval. Removing the zero frequency, where \(\T_i\chi_0=1\), gives
the sum over \(a\ne0\). The symmetry follows from
\(\chi_{-a}=\overline{\chi_a}\), and \(\|\T_i\chi_a\|_2\le1\) from
\(|\chi_a|=1\). For odd \(p\), the frequencies \(a\) and \(-a\) are
distinct when \(a\ne0\), so
\(2\|\T_i\chi_a\|_2^2\le\sum_{a'\ne0}\|\T_i\chi_{a'}\|_2^2\le Q-1\).
This is the first term of the minimum, and for \(Q=2\) it is the
whole bound.

Now suppose \(p\ge16Q\) and fix \(a\ne0\). Write
\(z_b=(\T_i\chi_a)(b)=|z_b|e^{i\phi_b}\), with \(\phi_b\) arbitrary
when \(z_b=0\). Since \(|z_b|\le1\),
\[
  \|\T_i\chi_a\|_2^2
  =\sum_b\nu_i(b)|z_b|^2
  \le\sum_b\nu_i(b)|z_b|
  =\frac1p\sum_{x\in\Fp}
   \operatorname{Re}\bigl(e^{-i\phi_{\tau_i(x)}}\chi_a(x)\bigr).
\]
The values \(\chi_a(x)\) are the \(p\) equally spaced points of the
unit circle. Around each of the at most \(Q\) angles \(\phi_b\) take
the arc of angular radius \(\pi/(4Q)\); each arc contains at most
\(p/(4Q)+1\) of the points, so their union contains at most
\(p/4+Q\le3p/8\) since \(p\ge16Q\). Thus at least half the points lie
outside every arc, where the real part above is at most
\(\cos(\pi/(4Q))\); elsewhere it is at most one. Hence
\[
  \|\T_i\chi_a\|_2^2
  \le\frac{1+\cos(\pi/(4Q))}{2}
  =\cos^2\frac{\pi}{8Q} .
\]
Finally, \(\sin u\ge2u/\pi\) for \(0\le u\le\pi/2\) gives
\(\sin^2(\pi/(8Q))\ge1/(16Q^2)\).
\end{proof}

From now on we assume \(p\ge16Q\), or \(Q=2\) and \(p\) odd. Note that
\(\eta_Q<1\) bounds individual characters, not every function of mean
zero: a function of \(\tau_i(x)\) alone is fixed by \(\T_i\), which is
why \cref{lem:component-bound} first reduces to a sum of products of
character norms.

\subsection{Applying list recovery}

We keep \(j\) of the \(k-1\) factors in \cref{lem:component-bound} and
bound the rest by \(\eta_Q\). An extremal argument then replaces the
retained weights by indicators of input lists, so the sum counts dual
words with at least \(j\) agreements.

\begin{lemma}[Replacing weights by list indicators]
\label{lem:weights-to-lists}
Let \(\Omega\) be a finite set and \(2\le Q\le|\Omega|+1\) an integer.
Let \(P(w_1,\ldots,w_n)\) be a nonnegative linear combination of terms
of the form
\[
  \prod_{i\in S}w_i(a_i),\qquad S\subseteq[n],\quad a_i\in\Omega .
\]
For all functions \(w_i\colon\Omega\to[0,1]\) with
\(\sum_{a\in\Omega}w_i(a)\le Q-1\), there are sets
\(V_i\subseteq\Omega\) of size \(Q-1\) such that
\[
  P(w_1,\ldots,w_n)\le P(\one_{V_1},\ldots,\one_{V_n}) .
\]
\end{lemma}

\begin{proof}
Fix all weights except \(w_i\). Each term contains at most one factor
from coordinate \(i\), so \(P\) has the form
\(c_0+\sum_{a\in\Omega}c_a w_i(a)\) with all \(c_a\ge0\). Subject to
\(0\le w_i(a)\le1\) and \(\sum_a w_i(a)\le Q-1\), this is maximized by
putting weight one on the \(Q-1\) largest coefficients and zero
elsewhere; let \(V_i\) be the set of those \(Q-1\) indices. Replacing
\(w_i\) by \(\one_{V_i}\) does not decrease \(P\). Repeating for each
coordinate in turn gives the claim.
\end{proof}

\begin{lemma}[A list-recovery bound for the components]
\label{lem:component-list-recovery}
For every integer \(1\le j<k\) with \(j\le n\),
\[
  \sum_{A\subseteq[n]\colon|A|\ge k}\|\T(f_C^{=A})\|_2^2
  \le\eta_Q^{\,k-1-j}\,\List_{C^\perp}(n-j,Q-1) .
\]
\end{lemma}

\begin{proof}
Let \(\sum_A\) range over supports of size at least \(k\). For each
such \(A\) fix \(I_A\subseteq A\) with \(|I_A|=k-1\) and
\(J_A\subseteq I_A\) with \(|J_A|=j\). Apply \cref{lem:component-bound}
with \(J=A\setminus I_A\). Every word with support \(A\) is nonzero at
each coordinate of \(A\), so \cref{lem:local-fourier-energy} bounds each
of the \(k-1-j\) factors indexed by \(I_A\setminus J_A\) by \(\eta_Q\):
\[
  \sum_A\|\T(f_C^{=A})\|_2^2
  \le\sum_A\sum_{\vec y\in C^\perp\colon\supp(\vec y)=A}
     \prod_{i\in I_A}\|\T_i\chi_{y_i}\|_2^2
  \le\eta_Q^{\,k-1-j}
     \sum_A\sum_{\vec y\in C^\perp\colon\supp(\vec y)=A}
     \prod_{i\in J_A}\|\T_i\chi_{y_i}\|_2^2 .
\]
Regard the last double sum as a function \(P\) of the weights
\(w_i(a)=\|\T_i\chi_a\|_2^2\), \(a\in\Fp^*\): it is a sum of monomials
\(\prod_{i\in J_A}w_i(y_i)\), one for each dual word, each with
coefficient one. \Cref{lem:weights-to-lists} applies to \(P\) with
\(\Omega=\Fp^*\), since \(Q-1\le p-1\) and
\cref{lem:local-fourier-energy} gives the constraints on the weights.
This yields lists \(V_i\subseteq\Fp^*\) of size \(Q-1\) with
\(P(\vec w)\le P(\one_{V_1},\ldots,\one_{V_n})\). With indicator
weights, a monomial equals one exactly when its dual word agrees with
the lists on all \(j\) coordinates of \(J_A\), and each dual word
occurs once because its support determines \(A\). With input lists
\(S_i=V_i\) for every \(i\), each counted word therefore has at least
\(j\) agreements, so
\[
  P(\one_{V_1},\ldots,\one_{V_n})
  =\sum_A\sum_{\vec y\in C^\perp\colon\supp(\vec y)=A}
     \prod_{i\in J_A}\one[y_i\in V_i]
  \le\List_{C^\perp}(n-j,Q-1) ,
\]
which proves the lemma.
\end{proof}

\begin{proof}[Proof of \cref{thm:list-recovery-transfer}]
Set \(k=R+1\) and \(j=n-R\). Then \(1\le j<k\), \(k-1-j=2R-n\), and
\(n-j=R\). \Cref{prop:full-l2-bound,lem:component-list-recovery} give
\[
  \SD\bigl(\vec\tau(\vec U_C),\vec\tau(\vec U_{\Fp^n})\bigr)
  \le\frac12\Bigl(\sum_{|A|\ge k}\|\T(f_C^{=A})\|_2^2\Bigr)^{1/2}
  \le\frac12\,\eta_Q^{\,R-n/2}\sqrt{\List_{C^\perp}(R,Q-1)} .
  \qedhere
\]
\end{proof}

\begin{corollary}[List decoding implies one-bit leakage resilience]
\label{thm:list-transfer}
Let \(p\) be an odd prime, \(C\subseteq\Fp^n\) a linear code, and
\(R\) an integer with \(n/2<R<n\) such that every nonzero word of
\(C^\perp\) has weight at least \(R+1\). Then, for every tuple of
one-bit maps,
\[
  \SD\bigl(\vec\tau(\vec U_C),\vec\tau(\vec U_{\Fp^n})\bigr)
  \le\frac12\,2^{\,n/2-R}\sqrt{\List_{C^\perp}(R)} .
\]
\end{corollary}

\begin{proof}
Take \(Q=2\) in \cref{thm:list-recovery-transfer}; then \(\eta_2=1/2\)
and \(\List_{C^\perp}(R,1)=\List_{C^\perp}(R)\).
\end{proof}

We now derive the above-half result for GRS codes from the transfer
and \cref{thm:rs-list-decoding}.

\begin{theorem}[Logarithmic leakage above half]
\label{thm:grs-constant-bits}
Fix \(0<\delta<1/2\). There are constants \(c_\delta>0\) and
\(n_0(\delta)\) such that, for \(n\ge n_0(\delta)\) and every prime
\(p>n\), every GRS code over \(\Fp\) of length \(n\) and rate at
least \(1/2+\delta\) is
\(\exp(-c_\delta n^{1-\delta/12})\)-resilient to
\(\floor{(\delta/24)\log_2n}\)-bit local leakage.
% More precisely, it is \(\exp(-c_\delta n/2^{2m})\)-resilient to
% \(m\)-bit local leakage for every integer
% \(1\le m\le(\delta/24)\log_2n\).
% Shamir sharing has the same guarantee whenever
% \(k\ge(1/2+\delta)n\). One can take \(c_\delta=\delta/128\).
\end{theorem}

\begin{proof}[Proof of \cref{thm:grs-constant-bits}]
Fix an integer \(1\le m\le(\delta/24)\log_2n\). Write
\(C=\GRS_{\vec\alpha,\vec v}(n,k-1)\); the full-space code is
immediate, so assume \(k-1<n\). We prove the code estimate under the
slightly weaker condition \(k\ge(1/2+\delta)n\), which covers both
codes of rate at least \(1/2+\delta\) and the zero-secret codes of the
Shamir schemes in the theorem. Throughout, \(n\) is sufficiently large
depending only on \(\delta\).

Put \(Q=2^m\le n^{\delta/24}\) and \(R=\floor{k-1-(\delta/2)n}\). Since
\(k\ge(1/2+\delta)n\), we have \(R-n/2\ge\delta n/2-2\ge\delta n/4\);
also \(R<k\) and \(p>n\ge16Q\). The dual code has dimension \(n-k+1\)
and minimum distance \(k\ge R+1\), so \cref{thm:list-recovery-transfer}
applies. We use the specialization of \cref{thm:rs-list-decoding} to
dual GRS codes stated after it, with gap \(\delta/2\) and input-list
size \(Q-1\): its hypotheses hold since \(k-1\ge\delta n/2\) and
\(c_{\delta/2}(Q-1)^{18/\delta}\le c_{\delta/2}n^{3/4}\le n\), and its
radius is exactly \(R\). Consequently
\[
  \log\List_{C^\perp}(R,Q-1)
  \le c_{\delta/2}\,Q^{6/\delta}\log n
  \le c_{\delta/2}\,n^{1/4}\log n .
\]
With \(\eta_Q\le1-1/(16Q^2)\), the transfer gives
\[
  \SD\bigl(\vec\tau(\vec U_C),\vec\tau(\vec U_{\Fp^n})\bigr)
  \le\frac12\exp\Bigl(-\frac{\delta n}{64Q^2}
     +\frac{c_{\delta/2}}2\,n^{1/4}\log n\Bigr)
  \le\frac12\exp\Bigl(-\frac{\delta n}{128Q^2}\Bigr),
\]
where the last step holds for all sufficiently large \(n\), because
\(Q^2\le n^{\delta/12}\) and \(\delta/12+1/4<1\) give
\(c_{\delta/2}Q^2n^{1/4}\log n\le\delta n/64\).
For Shamir sharing, apply
\cref{lem:shamir-from-code} to the zero-secret code; its factor two
cancels the prefactor \(1/2\). Thus \(c_\delta=\delta/128\) suffices in
both cases, and \(2^{2m}\le n^{\delta/12}\) gives the stated bound.
\end{proof}

\section{Structured leakage at constant-rate thresholds}
\label{sec:structured-leakage}

This section introduces the second tool, an iterated Cauchy--Schwarz
inequality of generalized von Neumann type tailored to MDS systems of
linear forms (\cref{lem:iterated-cs,cor:iterated-cs-count}). It bounds
an average of a product of functions along the forms of a GRS code by
a product of Gowers norms, and for an MDS system it retains linearly
many such norms at every constant rate. We apply it on its own, through
the \(L^1\) bound of \cref{prop:leakage-decomposition} with \(r_0=k-1\), to
prove \cref{thm:structured-leakage}: resilience to quadratic-residue
leakage and to almost all \(\Theta_c(\log_2p)\)-bit local leakage
functions at every positive constant rate.
\Cref{sec:main-result} combines the same inequality with the
list-decoding estimates of \cref{sec:list-transfer}.

Nguyen~\cite{Nguyen24} relates leakage security to Gowers norms and
bounds these norms for quadratic-residue and random leakage, but
retains a single Gowers-norm factor, so the bounds need an
exponentially large field at constant rate. The MDS von Neumann
inequality supplies linearly many factors from the same local
estimates, which allows fields of size just above \(n\).

Recall that \(\tau_{\mathrm{QR}}(x)=1\) if \(x\) is a square in \(\Fp\),
including zero, and \(\tau_{\mathrm{QR}}(x)=0\) otherwise.

\begin{theorem}[Structured leakage at constant rate]
\label{thm:structured-leakage}
Fix \(0<c<1\). There are constants \(\eta_c,a_c,b_c>0\) and
\(n_0(c)\) such that, for \(n\ge n_0(c)\) and every prime \(p>n\),
the following hold for GRS codes
\(C=\GRS_{\vec\alpha,\vec v}(n,k-1)\) with \(k\ge cn\).
\begin{enumerate}
\item For \(\tau_i(x)=\tau_{\mathrm{QR}}(a_ix+b_i)\), where
  \(a_i\in\Fp^*\) and \(b_i\in\Fp\),
  \[
    \SD\bigl(\vec\tau(\vec U_C),\vec\tau(\vec U_{\Fp^n})\bigr)
    \le\exp(-a_cn\log p).
  \]
\item For every integer \(1\le m\le\eta_c\log_2p\), all but an
  \(\exp(-b_c\sqrt p)\) fraction of tuples
  \(\vec\tau=(\tau_i\colon\Fp\to\{0,1\}^m)_{i\in[n]}\) satisfy
  \[
    \SD\bigl(\vec\tau(\vec U_C),\vec\tau(\vec U_{\Fp^n})\bigr)
    \le\exp(-a_cn\log p)
  \]
  simultaneously for every such code \(C\).
\end{enumerate}
Both conclusions hold for \((n,k)\)-Shamir sharing between every pair
of secrets, with the same error bounds.
\end{theorem}

The product bound for Gowers norms, together with standard
concentration estimates (\cref{lem:random-gowers-concentration}),
gives resilience to almost all local leakage functions in part~2.
This does not assert resilience to every local leakage function as in
\cref{def:code-leakage-resilience}.

\begin{remark}[Affine changes of the leakage input]
\label{rem:shifted-qr}
Replacing \(f(x)\) by \(f(ax+b)\) with \(a\ne0\) changes no Gowers
norm: substitute \(x\mapsto ax+b\) and \(h_t\mapsto ah_t\) in the
definition. This gives the affine form of quadratic-residue leakage in
\cref{thm:structured-leakage}, and likewise every tuple of random maps
that satisfies the Gowers-norm event below remains successful after
coordinatewise affine changes of input. The changes may depend on the
public parameters but are fixed before the codeword is sampled.
\end{remark}

For Shamir sharing with \(k\ge cn\), \cref{lem:shamir-from-code}
applied to \(C_0=\GRS_{\vec\alpha,\vec\alpha}(n,k-1)\), together with
the translation closure in \cref{rem:shifted-qr}, gives the code
bounds with an extra factor two; we choose the constant \(a_c\) below
small enough to absorb it.

\Cref{sec:gowers-leakage-transfer} turns a uniform bound on the local
Gowers norms into a leakage bound with explicit finite parameters;
\cref{sec:quadratic-leakage,sec:random-bit-leakage} supply the local
bounds for the two classes; \cref{sec:sqrt-threshold-field} extracts a
sufficient field size at square-root thresholds, where the attack
discussed at the end of that subsection shows that some field-size
condition is needed.

\subsection{A product bound for Gowers norms}
\label{sec:iterated-cs}

Iterated Cauchy--Schwarz is a standard tool in additive combinatorics and
higher-order Fourier analysis~\cite{Gowers01,TaoVu06,GreenTao10}.
We use the following consequence of Manners's generalized von Neumann
inequality for vector-valued linear maps~\cite[Section~3.1]{Manners21}.
Grouping several forms into one map gives a product of Gowers norms.
We include a self-contained proof of the form used below.

Fix a prime \(p\), integers \(k\ge2\) and \(r\ge1\), and write
\(V=\Fp^{k-1}\).

\begin{definition}[Linear forms and systems of forms]
  A \emph{linear form} on \(V\) is a map
  \(L(\vec z)=\sum_{a=1}^{k-1}c_a z_a\), with
  \(\vec c\in\Fp^{k-1}\). A \emph{system of linear forms} is an
  ordered tuple \((L_1,\ldots,L_r)\). A subfamily is independent if
  its coefficient vectors are linearly independent over \(\Fp\).
  The system is \emph{MDS} (maximum-distance-separable) if every
  subfamily of at most \(k-1\) forms is independent.

  An \emph{affine-linear form} is a map \(\psi=b+L\), with
  \(b\in\Fp\) and \(L\) linear. For a system
  \(\Psi=(\psi_i)_{i=1}^r\), where \(\psi_i=b_i+L_i\),
  independence and the MDS property always refer to the linear parts
  \(L_i\).
\end{definition}

\begin{definition}[Average along a system of forms]
  For functions \(\vec f=(f_1,\ldots,f_r)\), with
  \(f_i:\Fp\to\C\), define
  \[
    \Lambda_\Psi(\vec f)
      =\E_{\vec z\in V}\left[\prod_{i=1}^r f_i(\psi_i(\vec z))\right].
  \]
\end{definition}

For \(C=\GRS_{\vec\alpha,\vec v}(n,k-1)\), the forms
\(L_i(\vec z)=v_i\sum_{j=1}^{k-1}z_j\alpha_i^{j-1}\)
give a uniform codeword when \(\vec z\) is uniform in \(V\).
They form an MDS system by the Vandermonde determinant.
More generally, independent linear parts give independent uniform
evaluations, even after adding constants.
We use \(\sqcup\) for a disjoint union and regard the empty family
as independent.

\begin{lemma}[A product of Gowers norms]
  \label{lem:iterated-cs}
  Let \(\Psi=(b_i+L_i)_{i=1}^r\) be a system on \(V\), and let
  \(f_i:\Fp\to\C\). Suppose
  \([r]=J\sqcup I_1\sqcup\cdots\sqcup I_t\), with \(t\ge2\),
  and the linear parts indexed by \(I_a\cup J\) are independent
  for every \(a\). Then
  \[
    |\Lambda_\Psi(\vec f)|
      \le\prod_{j\in J}\norm{f_j}_{U^t}
          \prod_{i\in[r]\setminus J}\norm{f_i}_\infty.
  \]
\end{lemma}

\begin{proof}
The claim is immediate if \(J=\emptyset\) or some \(f_i\) is zero.
By homogeneity, assume \(\norm{f_i}_\infty\le1\) for every \(i\notin J\).
Replacing \(f_i(x)\) by \(f_i(b_i+x)\), which preserves both norms,
reduces to linear forms. Group the forms in \(J\) into one map and
their functions into one function, and do the same for each \(I_a\):
\[
\begin{aligned}
  L_J(\vec z)&=(L_j(\vec z))_{j\in J},&
  F(\vec x)&=\prod_{j\in J}f_j(x_j),\quad \vec x\in\Fp^J;\\
  L_{I_a}(\vec z)&=(L_i(\vec z))_{i\in I_a},&
  G_a(\vec y)&=\prod_{i\in I_a}f_i(y_i),\quad \vec y\in\Fp^{I_a}.
\end{aligned}
\]
Then \(\norm{G_a}_\infty\le1\) and
\(\Lambda_\Psi(\vec f)=\E_{\vec z}[F(L_J\vec z)
\prod_{a=1}^tG_a(L_{I_a}\vec z)]\).
Put \(K_a=\ker L_{I_a}=\bigcap_{i\in I_a}\ker L_i\), with
\(K_a=V\) if \(I_a=\emptyset\).
Independence of the forms indexed by \(I_a\cup J\) says that
\((L_J,L_{I_a}):V\to\Fp^J\times\Fp^{I_a}\) is surjective; hence
\[
  L_J(K_a)=\Fp^J\quad\text{for every }a,
  \qquad L_J(V)=\Fp^J.
\]

For \(0\le s\le t\), let \(\vec h_u,\vec h'_u\) be independent and
uniform in \(K_u\) for \(u\le s\), independent of the uniform
\(\vec z\in V\). Write \(\vec h=(\vec h_u)_{u\le s}\) and
\(\vec h'=(\vec h'_u)_{u\le s}\) for these collections.
For \(\vec\omega\in\{0,1\}^s\), write
\(\vec h_{\vec\omega}=\sum_{u\le s}\vec h_u^{(\omega_u)}\), where
\(\vec h_u^{(0)}=\vec h_u\) and \(\vec h_u^{(1)}=\vec h'_u\). Define
\[
  \Lambda_s=\E_{\vec z,\vec h,\vec h'}\Bigl[
    \prod_{\vec\omega\in\{0,1\}^s}\comp^{|\vec\omega|}
    \Bigl(F\bigl(L_J(\vec z+\vec h_{\vec\omega})\bigr)
      \prod_{a>s}G_a\bigl(L_{I_a}(\vec z+\vec h_{\vec\omega})\bigr)\Bigr)
    \Bigr],
\]
so that \(\Lambda_0=\Lambda_\Psi(\vec f)\). We claim that
\(|\Lambda_s|^2\le\Lambda_{s+1}\) for \(0\le s<t\), where
\(\Lambda_{s+1}\ge0\). Indeed, write the integrand of \(\Lambda_s\) as
\[
  \prod_{\vec\omega}\comp^{|\vec\omega|}
     G_{s+1}\bigl(L_{I_{s+1}}(\vec z+\vec h_{\vec\omega})\bigr)
       \cdot\Phi(\vec z,\vec h,\vec h'),
\]
where \(\Phi\) collects the factors involving \(F\) and \(G_a\) for
\(a>s+1\). Translating \(\vec z\) by an independent uniform
\(\vec h_{s+1}\in K_{s+1}\) preserves its distribution and leaves the
\(G_{s+1}\) factors unchanged. Their product has modulus at most one,
so Cauchy--Schwarz in \((\vec z,\vec h,\vec h')\) gives
\[
\begin{aligned}
  |\Lambda_s|^2
  &\le\E_{\vec z,\vec h,\vec h'}
       \left|\E_{\vec h_{s+1}}
         [\Phi(\vec z+\vec h_{s+1},\vec h,\vec h')]\right|^2\\
  &=\E_{\vec z,\vec h,\vec h',\vec h_{s+1},\vec h'_{s+1}}\Bigl[
       \Phi(\vec z+\vec h_{s+1},\vec h,\vec h')
       \overline{\Phi(\vec z+\vec h'_{s+1},\vec h,\vec h')}\Bigr]
    =\Lambda_{s+1}.
\end{aligned}
\]
The last integrand is exactly that of \(\Lambda_{s+1}\): the new
coordinate \(\omega_{s+1}\) selects \(\vec h_{s+1}\) or
\(\vec h'_{s+1}\) and contributes one conjugation.

Iterating gives
\[
  |\Lambda_\Psi(\vec f)|^{2^t}\le\Lambda_t
  =\E_{\vec z,\vec h,\vec h'}\Bigl[
    \prod_{\vec\omega\in\{0,1\}^t}
    \comp^{|\vec\omega|}F\bigl(L_J(\vec z+\vec h_{\vec\omega})\bigr)
    \Bigr].
\]
Put \(\vec d_u=\vec h'_u-\vec h_u\), so that
\(\vec z+\vec h_{\vec\omega}
=\vec z+\sum_u\vec h_u+\sum_u\omega_u\vec d_u\),
and set \(\vec x=L_J(\vec z+\sum_u\vec h_u)\) and
\(\vec y_u=L_J(\vec d_u)\).
Surjectivity of \(L_J\) makes \(\vec x\) uniform on \(\Fp^J\) and
independent of \((\vec h,\vec h')\). Each \(\vec d_u\) is uniform on
\(K_u\), so
\(L_J(K_u)=\Fp^J\) makes \(\vec y_u\) uniform on \(\Fp^J\), and all
these variables are independent. Therefore
\[
  \Lambda_t=\E_{\vec x,\vec y_1,\ldots,\vec y_t\in\Fp^J}\Bigl[
    \prod_{\vec\omega\in\{0,1\}^t}\comp^{|\vec\omega|}
      F\Bigl(\vec x+\sum_{u=1}^t\omega_u\vec y_u\Bigr)\Bigr]
    =\norm F_{U^t(\Fp^J)}^{2^t}.
\]
Finally, the definition of the Gowers norm factors over coordinates:
\[
\begin{aligned}
  \norm F_{U^t(\Fp^J)}^{2^t}
    =\E_{\vec x,\vec h_1,\ldots,\vec h_t\in\Fp^J}
       \left[\prod_{\vec\omega\in\{0,1\}^t}\prod_{j\in J}
       \comp^{|\vec\omega|}
       f_j\left(x_j+\sum_{u=1}^t\omega_u h_{u,j}\right)\right]
    =\prod_{j\in J}\norm{f_j}_{U^t}^{2^t}.
\end{aligned}
\]
Taking \(2^t\)-th roots and restoring the original scaling proves
the claim.
\end{proof}

For an MDS system, the lemma's independence conditions reduce to
\(|I_a|+|J|\le k-1\). This is the only use of the MDS property in
the following corollary.

\begin{corollary}[Number of \(U^t\) factors]
  \label{cor:iterated-cs-count}
  Let \(\Psi\) be an MDS system, let \(r\ge k\) and \(t\ge2\), and define
  \[
    q_t(r,k)=\floor{\frac{t(k-1)-r}{t-1}}.
  \]
  If \(q_t(r,k)\ge1\), then, for every \(J\subseteq[r]\) of size
  \(q_t(r,k)\),
  \[
    |\Lambda_\Psi(\vec f)|
      \le\prod_{j\in J}\norm{f_j}_{U^t}
          \prod_{i\in[r]\setminus J}\norm{f_i}_\infty.
  \]
\end{corollary}

\begin{proof}
Put \(q=q_t(r,k)\). Since \(r\ge k\), we have \(q\le k-2<r\).
The definition of \(q\) gives \(r-q\le t(k-1-q)\). Thus we can
partition the indices outside \(J\) into \(t\) groups \(I_a\), each
of size at most \(k-1-q\). Every \(I_a\cup J\) has size at most
\(k-1\), so \cref{lem:iterated-cs} applies.
\end{proof}

If the functions are bounded by one and their \(U^t\) norms are at
most \(\rho<1\), this gives \(\rho^{q_t(r,k)}\).
For structured leakage, we allow \(t\) to grow with \(n\).
In \cref{sec:main-result}, we will reuse the bound with \(t=3\)
for the below-half theorem.

\subsection{From Gowers norms to leakage resilience}
\label{sec:gowers-leakage-transfer}

Let \(\tau_i\colon\Fp\to\{0,1\}^m\), and put
\(f_{i,b}(x)=\one[\tau_i(x)=b]-\nu_i(b)\), as in
\cref{lem:output-indicator-components}. The \(L^1\) bound of
\cref{prop:leakage-decomposition} with \(r_0=k-1\) reduces
the leakage distance to averages of products of these functions.

\begin{proposition}[Finite-parameter leakage bound]
\label{prop:gowers-leakage-transfer}
Let \(C\) be an \([n,k-1,n-k+2]_p\) MDS code, and let \(t\ge2\) be an
integer with \(q_t(n,k)\ge1\). If \(\norm{f_{i,b}}_{U^t}\le\rho\le1\)
for all \(i\in[n]\) and \(b\in\{0,1\}^m\), then
\[
  \SD\bigl(\vec\tau(\vec U_C),\vec\tau(\vec U_{\Fp^n})\bigr)
  \le\frac12(1+2^m)^n\rho^{q_t(n,k)} .
\]
For one-bit maps, if the functions \(f_i=2(\tau_i-\E_x[\tau_i(x)])\)
satisfy \(\norm{f_i}_{U^t}\le\rho\le1\), the distance is at most
\(3^n\rho^{q_t(n,k)}\).
\end{proposition}

\begin{proof}
The dual of \(C\) has minimum distance \(k\), so
\cref{prop:leakage-decomposition,lem:output-indicator-components} give
\[
  \SD\bigl(\vec\tau(\vec U_C),\vec\tau(\vec U_{\Fp^n})\bigr)
  \le\frac12\sum_{|A|\ge k}\sum_{\vec b_A\in(\{0,1\}^m)^A}
       \left|\E_{\vec X\sim U_C}
       \left[\prod_{i\in A}f_{i,b_i}(X_i)\right]\right| .
\]
Fix \(A\) with \(|A|=r\ge k\). The coordinate forms of a linear
parametrization of \(C\), restricted to \(A\), form an MDS system,
since any \(k-1\) of them are independent,
and \(\norm{f_{i,b_i}}_\infty\le1\). Since
\(q_t(r,k)\ge q_t(n,k)\), \cref{cor:iterated-cs-count} bounds each
average by \(\rho^{q_t(r,k)}\le\rho^{q_t(n,k)}\). Hence
\[
  \SD\bigl(\vec\tau(\vec U_C),\vec\tau(\vec U_{\Fp^n})\bigr)
  \le\frac12\rho^{q_t(n,k)}\sum_{r=k}^n\binom nr2^{mr}
  \le\frac12(1+2^m)^n\rho^{q_t(n,k)}.
\]
For one-bit maps, \(f_{i,1}=f_i/2\) and \(f_{i,0}=-f_i/2\), so the
first bound gives \(\frac12\,3^n(\rho/2)^{q_t(n,k)}\le
3^n\rho^{q_t(n,k)}\), which completes the proof.
\end{proof}

At constant rate, fix an integer \(t\ge2\) with \(tc>1\) and put
\(\beta=(tc-1)/(t-1)>0\). For \(k\ge cn\), the number of retained
Gowers factors satisfies
\begin{equation}
\label{eq:constant-rate-survivors}
  q_t(n,k)=\floor{\frac{t(k-1)-n}{t-1}}
  \ge\beta n-\frac t{t-1}-1\ge\beta n-3 .
\end{equation}
Thus local norms of size \(p^{-\Omega_c(1)}\) give
\(\exp(-\Omega_c(n\log p))\) leakage distance even when
\(m\) is a sufficiently small constant times \(\log_2p\).

\subsection{Quadratic-residue leakage}
\label{sec:quadratic-leakage}

Let \(\chi\colon\Fp\to\{-1,0,1\}\) be the quadratic character, with
\(\chi(0)=0\). The character-sum estimate of~\cite[Lemma~2]{Nguyen24}
gives \(\norm\chi_{U^t}^{2^t}\le2^t/\sqrt p\). For
\(\tau=\tau_{\mathrm{QR}}\) we have \(2\E_x[\tau(x)]=1+1/p\), so the
centered function is \(f=\chi+\one_{\{0\}}-1/p\). The Gowers product
of \(\one_{\{0\}}\) is one exactly when \(x=h_1=\dots=h_t=0\), so
\(\norm{\one_{\{0\}}}_{U^t}=p^{-(t+1)/2^t}\), and the triangle
inequality gives
\[
  \norm f_{U^t}
  \le\Bigl(\frac{2^t}{\sqrt p}\Bigr)^{1/2^t}+p^{-(t+1)/2^t}+\frac1p
  \le4p^{-1/2^{t+1}} .
\]
The same bound holds if the output at zero is changed to \(0\), where
the centered function is \(\chi-\one_{\{0\}}+1/p\).

\begin{proof}[Proof of \cref{thm:structured-leakage}, part~1]
Fix \(t\) and \(\beta\) as in \eqref{eq:constant-rate-survivors}, and
write \(q=q_t(n,k)\). For
\(n\) sufficiently large, \(q\ge1\) and \(4p^{-1/2^{t+1}}\le1\) for
every \(p>n\), and the affine changes of input preserve the local
bound by \cref{rem:shifted-qr}. \Cref{prop:gowers-leakage-transfer}
bounds the distance by
\[
  3^n\bigl(4p^{-1/2^{t+1}}\bigr)^q
  =\exp\Bigl(n\log3+q\log4-\frac{q}{2^{t+1}}\log p\Bigr).
\]
Since \(\beta n-3\le q\le n\) and \(p>n\), the exponent is at most
\(-a_1n\log p\) with \(a_1=\beta/2^{t+3}\), for \(n\) sufficiently
large depending only on \(c\). The bound is uniform over the code and
the affine changes of input.
\end{proof}

\subsection{Almost all logarithmic-bit leakage functions}
\label{sec:random-bit-leakage}

For a random map \(\tau\colon\Fp\to\{0,1\}^m\), each output indicator
minus \(2^{-m}\) has independent centered values. The cube-collision
argument for random signs in~\cite[Lemma~4]{Nguyen24} therefore applies
to these indicators as well. Subtracting their empirical means costs
at most a factor two in the Gowers norm; no balance assumption on the
output probabilities is needed.

\begin{lemma}[Gowers norm of a random leakage function]
\label{lem:random-gowers-concentration}
Let \(t\ge2\), \(m\ge1\), and \(0<\rho\le1\). For a uniformly random
map \(\tau\colon\Fp\to\{0,1\}^m\) and a fixed output \(b\), put
\(f_b(x)=\one[\tau(x)=b]-\Prb_{u\sim U_{\Fp}}[\tau(u)=b]\).
If \(p(\rho/2)^{2^t}\ge2\binom{2^t}2\), then
\[
  \Prb_\tau\bigl[\norm{f_b}_{U^t}>\rho\bigr]
  \le\exp\left(-\frac{p(\rho/2)^{2^{t+1}}}{2^{2t+3}}\right).
\]
\end{lemma}

\begin{proof}
Put \(\xi_b(x)=\one[\tau(x)=b]-2^{-m}\) and
\(T(\tau)=\norm{\xi_b}_{U^t}^{2^t}\). If all \(2^t\) vertices of a
Gowers cube are distinct, independence and \(\E_\tau[\xi_b(x)]=0\)
make its expected product zero. Each pair of vertices coincides with
probability \(1/p\), and \(|\xi_b|\le1\), so
\[
  \E_\tau[T]\le\frac{\binom{2^t}2}{p},\qquad
  \norm{f_b}_{U^t}
  =\norm{\xi_b-\E_x[\xi_b(x)]}_{U^t}
  \le\norm{\xi_b}_{U^t}+|\E_x[\xi_b(x)]|
  \le2\norm{\xi_b}_{U^t}.
\]
Changing one value \(\tau(a)\) affects at most a \(2^t/p\) fraction
of the cube products, each by at most two. Thus \(T\) changes by at
most \(2^{t+1}/p\), and bounded differences~\cite{McDiarmid89} gives
\[
  \Prb_\tau[T-\E_\tau[T]\ge u]
  \le\exp\left(-\frac{pu^2}{2^{2t+1}}\right).
\]
The hypothesis gives \(\E_\tau[T]\le\frac12(\rho/2)^{2^t}\), whereas
\(\norm{f_b}_{U^t}>\rho\) implies \(T>(\rho/2)^{2^t}\). Taking
\(u=\frac12(\rho/2)^{2^t}\) proves the claim.
\end{proof}

\begin{proof}[Proof of \cref{thm:structured-leakage}, part~2]
Fix \(t\) and \(\beta\) as in \eqref{eq:constant-rate-survivors}, and
choose
\[
  \rho=p^{-1/2^{t+2}},\qquad
  \eta_c=\frac{\beta}{2^{t+4}},\qquad
  1\le m\le\eta_c\log_2p.
\]
For large \(p\), the hypothesis of
\cref{lem:random-gowers-concentration} holds because
\(p(\rho/2)^{2^t}=2^{-2^t}p^{3/4}\). Its conclusion is
\(\Prb[\norm{f_{i,b}}_{U^t}>\rho]\le\exp(-b_0\sqrt p)\) for a
constant \(b_0>0\) depending only on \(t\). A union bound over the
\(n2^m\) pairs \((i,b)\) gives
\[
  \Prb[\norm{f_{i,b}}_{U^t}>\rho\text{ for some }(i,b)]
  \le n2^m\exp(-b_0\sqrt p)
  \le p^{1+\eta_c}\exp(-b_0\sqrt p)
  \le\exp(-b_0\sqrt p/2).
\]
On the complementary event, take \(n\) large enough that
\(q_t(n,k)\ge3\beta n/4\). Since
\(\log(1+2^m)\le\log2+\eta_c\log p\),
\cref{prop:gowers-leakage-transfer} gives
\[
\begin{aligned}
  \SD\bigl(\vec\tau(\vec U_C),\vec\tau(\vec U_{\Fp^n})\bigr)
  &\le\frac12\exp\left(n\log(1+2^m)
                   -\frac{q_t(n,k)}{2^{t+2}}\log p\right)\\
  &\le\frac12\exp\left(n\log2-2\eta_cn\log p\right)
   \le\frac12\exp(-\eta_cn\log p).
\end{aligned}
\]
The event depends only on the maps, so this holds simultaneously for
all codes in the theorem and, by \cref{rem:shifted-qr}, after all
coordinatewise affine changes of input. Take \(b_c=b_0/2\) and
\(a_c=\min\{a_1/2,\eta_c\}\), increasing \(n_0(c)\) if needed.
The factor two in the Shamir reduction is absorbed by the displayed
factor \(1/2\) here and by the slack in \(a_1\) for part~1.
\end{proof}

For fixed \(m\), the same proof gives the stronger success probability
\(1-\exp(-\Omega_{c,m}(p))\) with leakage distance
\(\exp(-\Omega_{c,m}(n))\): choose a constant \(\rho>0\) such that
\(\beta\log(1/\rho)>\log(1+2^m)\) and apply the same two bounds.

\subsection{Field size for square-root thresholds}
\label{sec:sqrt-threshold-field}

When \(k=\Theta(\sqrt n)\), the Cauchy--Schwarz order \(t\) must grow
with \(n\), so we use the finite-parameter bound rather than the
constant-rate theorem. The next corollary gives an explicit sufficient
field size for leakage error \(\exp(-n)\) in both applications.

\begin{corollary}[An explicit field-size bound]
\label{cor:sqrt-threshold-field}
Let \(3\le k\le n\) and \(t=\ceil{(n-1)/(k-2)}\). If
\(\log p\ge8n\,2^t\), then for every
\(C=\GRS_{\vec\alpha,\vec v}(n,k-1)\),
\(\SD(\vec\tau(\vec U_C),\vec\tau(\vec U_{\Fp^n}))\le\exp(-n)\) for
quadratic-residue leakage with the affine input changes of
\cref{thm:structured-leakage}; and for independent uniformly random
Boolean maps, with probability at least \(1-\exp(-n)\), the same bound
holds simultaneously for every such code. Both conclusions hold for
Shamir leakage between every pair of secrets with the same error.
\end{corollary}

\begin{proof}
By the choice of \(t\), \(t(k-2)\ge n-1\), i.e.\ \(t(k-1)-n\ge t-1\),
so \(q=q_t(n,k)\ge1\). Put \(\rho=\exp(-n)/(2\cdot3^n)\) and
\(f_i=2(\tau_i-\E_x[\tau_i(x)])\). If
\(\norm{f_i}_{U^t}\le\rho\) for every \(i\), then
\cref{prop:gowers-leakage-transfer} gives
\(3^n\rho^q\le3^n\rho=\frac12\exp(-n)\), and the factor \(\frac12\)
yields the Shamir error through \cref{lem:shamir-from-code}.

For quadratic-residue leakage, \cref{sec:quadratic-leakage} gives
\(\norm{f_i}_{U^t}\le4p^{-1/2^{t+1}}\le\rho\) whenever
\(\log p\ge2^{t+1}\log(4/\rho)=2^t\bigl(2\log8+2(1+\log3)n\bigr)\),
which follows from \(\log p\ge8n2^t\) since \(n\ge3\).

For a random Boolean map, \(2\tau-1\) has independent uniform signs.
The same cube-collision argument as in
\cref{lem:random-gowers-concentration} gives
\(\E_\tau[\norm{2\tau-1}_{U^t}^{2^t}]\le\binom{2^t}2/p\), and
\(\norm{2(\tau-\E_x[\tau(x)])}_{U^t}\le2\norm{2\tau-1}_{U^t}\).
Markov's inequality and a union bound therefore give
\[
  \Prb\bigl[\norm{f_i}_{U^t}>\rho\text{ for some }i\in[n]\bigr]
  \le\frac{n\,2^{2^t}\binom{2^t}2}{p\,\rho^{2^t}}
  \le\exp(-n),
\]
the last step because
\(2^t\log(2/\rho)+n+\log n+\log\binom{2^t}2
 \le2^t\bigl((1+\log3)n+2\log2\bigr)+n+\log n+2t\log2
 \le8n2^t\le\log p\),
using \(2\le t\le n-1\) and \(n+\log n+2t\log2\le4n\le n2^t\). The
event depends only on the maps and is preserved by affine input
changes, so the bound holds simultaneously for all codes and input
changes in the statement.
\end{proof}

For \(k=\Theta(\sqrt n)\) we have \(t=n/k+O(1)\), so the condition
reads \(\log p\ge Cn\,2^{n/k}\) for a constant \(C\). Concretely, for
\(n=k^2\) with \(k\ge5\), one gets \(t=k+3\) and \(q=1\),
so \(\log p\ge64n\,2^{\sqrt n}\) suffices; a single Gowers factor
carries the bound. This field size is doubly exponential in
\(\sqrt n\); we do not claim it is necessary.

\paragraph{A contrasting attack at square-root thresholds.}
Kopparty's algorithm for recovering a monic squarefree polynomial from
noisy quadratic-character values~\cite{Kopparty26} attacks Shamir
sharing when \(n=p-1\) and \(k\) is at most a small constant times
\(\sqrt n\): party \(i\) leaks
\(\tau_{\mathrm{QR}}(P(\alpha_i)+\alpha_i^k)\), which is the character
value of the monic polynomial \(Q(X)=X^k+P(X)\) at \(\alpha_i\) up to
at most \(k+1\) discrepancies (the roots of \(Q\) and the missing
point \(0\)). The decoder tolerates these in this degree range and
recovers \(Q\), hence the secret \(Q(0)\), whenever \(Q\) is
squarefree. For each fixed secret \(s\), the discriminant of \(Q\) is a
nonzero polynomial of degree at most \(2k-2\) in the random
coefficients of \(P\) (it is nonzero at \(X^k+s\) for \(s\ne0\) and at
\(X^k+X\) for \(s=0\), both squarefree since \(k<p\)), so \(Q\) fails
to be squarefree with probability at most \((2k-2)/p\). The attack
uses only the fixed input shifts covered by
\cref{thm:structured-leakage}, not additional leaked bits. At every
positive constant rate our theorem gives security even for \(p=n+1\);
at square-root thresholds the corollary above gives security over
sufficiently large fields. The intermediate range is open.

\section{Leakage resilience below half}
\label{sec:main-result}

\begin{theorem}[GRS codes below half]
\label{thm:grs-below-half}
There are absolute constants \(c_{\mathrm{bh}}<1/2\), \(\gamma>0\),
and \(n_0\) such that, for \(n\ge n_0\) and every prime \(p>n\),
every GRS code over \(\Fp\) of length \(n\) and rate at least
\(c_{\mathrm{bh}}\) is \(e^{-\gamma n}\)-resilient to one-bit local
leakage.
\end{theorem}

We prove \cref{thm:grs-below-half} by combining the three tools:
\cref{prop:master-inequality} with an intermediate cutoff \(r_0\) close
to \(n\); a supportwise form of the list-recovery argument of
\cref{sec:list-transfer}, which controls the squared components below
the cutoff (\cref{sec:lower}); and the MDS von Neumann inequality of
\cref{sec:iterated-cs} together with a second-moment bound, which
control the components above it (\cref{sec:high}). The supportwise
form is needed because the list bound of \cref{thm:rs-list-decoding}
is available only at radius \(R<k\), and below rate half
\cref{lem:component-list-recovery} cannot use it. Both estimates have
negative exponents at rate half; continuity gives a fixed rate
below half (\cref{sec:splice}). Throughout, \(C\) is a GRS code of
dimension \(k-1\) unless a statement says otherwise, and \(g_i,f_i\)
are as in \cref{sec:expansion}.

\subsection{List recovery for small supports}
\label{sec:lower}

\Cref{lem:component-list-recovery} needs the list bound at radius
\(n-j\) with \(j\le k-1\), whereas \cref{thm:rs-list-decoding} supplies
it at radius \(R<k\); below rate half, \(n-k+1>k-1\), so no
admissible \(j\) satisfies \(n-j\le R\). At a fixed support size \(r\),
however, every dual word has \(n-r\) zero coordinates, which count as
agreements with the list \(\{0\}\). Retaining \(j\) Fourier factors on
a common set \(H\) and putting \(\{0\}\) outside \(H\) brings the
radius down to \(r-j\), so \(j=r-R\) suffices, and \(j\le k-1\) as long
as \(r\le k+R-1\). Summing over \(H\) gives the following bound.

\begin{lemma}[List recovery for a fixed support size]
\label{lem:fixed-support-list-recovery}
In the setting of \cref{sec:list-transfer}, so that \(C^\perp\) has
minimum distance at least \(k\), the maps \(\tau_i\) have at most \(Q\)
output symbols, and \(p\ge16Q\) or \(Q=2\) with \(p\) odd, for all
integers \(k\le r\le n\) and \(1\le j<k\),
\[
  \sum_{|A|=r}\|\T(f_C^{=A})\|_2^2
  \le\eta_Q^{\,k-1-j}\,\List_{C^\perp}(r-j,Q-1)\,
     \frac{\binom nj}{\binom rj}.
\]
\end{lemma}

\begin{proof}
Fix \(H\subseteq[n]\) with \(|H|=j\). For each \(A\supseteq H\) of
size \(r\), choose \(I_A\) with \(H\subseteq I_A\subseteq A\) and
\(|I_A|=k-1\), which is possible since \(j\le k-1<r\). Apply
\cref{lem:component-bound} with \(J=A\setminus I_A\), and bound the
\(k-1-j\) factors indexed by \(I_A\setminus H\) by \(\eta_Q\) using
\cref{lem:local-fourier-energy}:
\[
  \sum_{\substack{A\supseteq H\\|A|=r}}\|\T(f_C^{=A})\|_2^2
  \le\sum_{\substack{A\supseteq H\\|A|=r}}
     \sum_{\substack{\vec y\in C^\perp\\\supp(\vec y)=A}}
       \prod_{i\in I_A}\|\T_i\chi_{y_i}\|_2^2
  \le\eta_Q^{\,k-1-j}
     \sum_{\substack{\vec y\in C^\perp\\
                     \wt(\vec y)=r,\;H\subseteq\supp(\vec y)}}
       \prod_{i\in H}\|\T_i\chi_{y_i}\|_2^2 .
\]
As in \cref{lem:component-list-recovery}, the last sum is a nonnegative
linear combination of monomials \(\prod_{i\in H}w_i(y_i)\) in the
weights \(w_i(a)=\|\T_i\chi_a\|_2^2\), \(a\in\Fp^*\).
The constraints on the weights follow from \cref{lem:local-fourier-energy}.
\Cref{lem:weights-to-lists} (the coordinates outside \(H\)
play no role) yields lists \(V_i\subseteq\Fp^*\) of size \(Q-1\),
\(i\in H\), with
\[
  \sum_{\substack{\vec y\in C^\perp\\
                  \wt(\vec y)=r,\;H\subseteq\supp(\vec y)}}
    \prod_{i\in H}\|\T_i\chi_{y_i}\|_2^2
  \le\bigl|\{\vec y\in C^\perp:\wt(\vec y)=r,\;
             y_i\in V_i\text{ for all }i\in H\}\bigr| .
\]
Take input lists \(S_i=V_i\) on \(H\) and \(S_i=\{0\}\) outside \(H\).
Every counted word agrees with the lists on \(H\) and on its \(n-r\)
zero coordinates, which lie outside \(H\); so the count is at most
\(\List_{C^\perp}(r-j,Q-1)\).

Finally, sum over \(H\). Each support of size \(r\) contains exactly
\(\binom rj\) sets \(H\), so
\[
  \binom rj\sum_{|A|=r}\|\T(f_C^{=A})\|_2^2
  =\sum_{|H|=j}\sum_{\substack{A\supseteq H\\|A|=r}}
     \|\T(f_C^{=A})\|_2^2
  \le\binom nj\,\eta_Q^{\,k-1-j}\List_{C^\perp}(r-j,Q-1) .
\]
Dividing by \(\binom rj\) proves the claim.
\end{proof}

For one-bit leakage the input lists have size one, so list recovery
becomes list decoding.

\begin{corollary}[List decoding for a fixed support size]
\label{lem:list-decoding-energy}
Let \(p\) be an odd prime and \(C\subseteq\Fp^n\) a linear code whose
dual has minimum distance at least \(k\). For all integers
\(k\le r\le n\) and \(R\) with \(1\le r-R<k\),
\[
  \sum_{|A|=r}|\Gamma_C(\vec g_A)|^2
  \le\List_{C^\perp}(R)\,2^{\,r-R-k+1}\,
     \frac{\binom n{r-R}}{\binom r{r-R}} .
\]
\end{corollary}

\begin{proof}
Take \(Q=2\) and \(j=r-R\) in \cref{lem:fixed-support-list-recovery}.
Then \(\eta_2=1/2\), \(\List_{C^\perp}(r-j,1)=\List_{C^\perp}(R)\),
and \(\|\T(f_C^{=A})\|_2^2=|\Gamma_C(\vec g_A)|^2\) by
\cref{lem:binary-components}.
\end{proof}

Write \(D=C^\perp\) for the remainder of this subsection. For our
decoding radius \(R<k\), the condition \(1\le r-R<k\) requires
\(k\le r\le\min\{n,k+R-1\}\); below rate half this leaves the
largest supports untreated, and \cref{sec:high} bounds those by
Cauchy--Schwarz. The next lemma gives a uniform estimate below the
cutoff \((1-\varepsilon)n\).

\begin{lemma}[Bound for smaller supports]
\label{lem:elementary-low-band}
Fix \(0<\zetaGap<c<1/2\) and \(0<\varepsilon<1\) with
\(1-\varepsilon<2c-\zetaGap\). Let \(k=k_n\) with
\((k_n-1)/n\to c\), and \(R=\floor{k-1-\zetaGap n}\). Assume
\(\List_D(R)\le L_n\) with \(\log L_n=o(n)\). Then every support size
\(k\le r\le(1-\varepsilon)n\) satisfies
\[
  \sum_{|A|=r}|\Gamma_C(\vec g_A)|^2
  \le\exp\bigl(nG(c,\zetaGap,\varepsilon)+o(n)\bigr),
\]
uniformly in the support size and the leakage maps, where
\[
  G(c,\zetaGap,\varepsilon)
  =(1-2c+\zetaGap)\log2
   +\log\frac1{2(c-\zetaGap)}-\frac{\varepsilon^2}{2}.
\]
In particular, if \(G(c,\zetaGap,\varepsilon)<0\), the sum is
exponentially small at each of these support sizes.
\end{lemma}

\begin{proof}
We first obtain an estimate with the integer parameters explicit.
Suppose \(\List_D(R)\le L\) with \(1\le R<k\). For
\(k\le r\le\min\{n,k+R-1\}\), put \(j=r-R\); then \(1\le j<k\) and
\(r-j=R\), so \cref{lem:list-decoding-energy} gives
\[
  \sum_{|A|=r}|\Gamma_C(\vec g_A)|^2
  \le L\,2^{j-k+1}\frac{\binom nj}{\binom rj} .
\]
Set \(t=n-r\) and \(g=n-k-R+1\). The binomial ratio has the exact
product form
\[
  2^{j-k+1}\frac{\binom nj}{\binom rj}
  =2^g\prod_{u=0}^{t-1}\frac{n-u}{2(R+t-u)} .
\]
For \(0\le u<t\) we have \(R+t-u\ge R+1\), and \(1-v\le\exp(-v)\)
with \(R+1\le n\) gives \(R/(R+1)\le\exp(-1/(R+1))\le\exp(-1/n)\).
Hence
\[
  \frac{n-u}{2(R+t-u)}
  \le\frac n{2R}\Bigl(1-\frac un\Bigr)\frac R{R+1}
  \le\frac n{2R}\exp\Bigl(-\frac{u+1}n\Bigr),
\]
and multiplying over \(u\) gives the finite bound
\[
  \sum_{|A|=r}|\Gamma_C(\vec g_A)|^2
  \le L\,2^g\Bigl(\frac n{2R}\Bigr)^{t}
     \exp\Bigl(-\frac{t^2}{2n}\Bigr),
\]
since \(t(t+1)/(2n)\ge t^2/(2n)\). Empty products cause no difficulty
when \(t=0\).

Now take \(R=\floor{k-1-\zetaGap n}\) and \(L=L_n\). For large \(n\),
\(1\le R<k\) and \(2R<n\), since \(R/n\to c-\zetaGap<1/2\).
Moreover,
\(k+R-1=(2c-\zetaGap)n+o(n)>(1-\varepsilon)n\), so
\(r\le k+R-1\) throughout the stated range. Using
\(\varepsilon n\le t\le n\) in the finite bound gives
\[
\begin{aligned}
  \sum_{|A|=r}|\Gamma_C(\vec g_A)|^2
  &\le L_n\,2^{n-k-R+1}\Bigl(\frac n{2R}\Bigr)^n
       \exp\Bigl(-\frac{\varepsilon^2n}{2}\Bigr)\\
  &=\exp\bigl(nG(c,\zetaGap,\varepsilon)+o(n)\bigr),
\end{aligned}
\]
where the equality uses \((k-1)/n\to c\), \(R/n\to c-\zetaGap\), and
\(\log L_n=o(n)\). The bound is independent of the support size and
the leakage maps.
\end{proof}

\subsection{Iterated Cauchy--Schwarz for large supports}
\label{sec:high}

We bound the second term of \eqref{eq:master-inequality}. The proof
separates two cases: Cauchy--Schwarz controls leakage
functions with sufficiently large total squared bias, while the
product of \(U^3\) norms controls the remaining functions. Put
\[
  \kappa_0=(3/4)^{1/8}<1,\qquad
  \beta_0=\frac{(\log\kappa_0)^2}{64}.
\]

\begin{proposition}[Bound for large supports]
\label{prop:qual-high-sum}
\label{prop:qual-high-bound}
Suppose \(2k-2\le n<3k-3\), put \(c=(k-1)/n\), and fix
\(0<\varepsilon<1/2\). Let \(r_0=\floor{(1-\varepsilon)n}\), and define
\[
  H(c,\varepsilon)
  =h(\varepsilon)+\max\Bigl\{
     \frac{3c-1}{2}\log\kappa_0+\sqrt{\beta_0},\;
     (1-2c)(\log2+1/2)-\frac{\beta_0}2
   \Bigr\}.
\]
Then
\[
  \frac12\sum_{A\subseteq[n]\colon|A|>r_0}|\Gamma_C(\vec f_A)|
  \le\exp\bigl(nH(c,\varepsilon)+O(\log n)\bigr).
\]
In particular,
\(H(1/2,0)=\max\{\frac18\log\kappa_0,-\beta_0/2\}<0\), with the
formula extended continuously to \(\varepsilon=0\).
\end{proposition}

The strict inequality persists when \(c\) is slightly below half
and \(\varepsilon\) is sufficiently small. We first establish the local
\(U^3\) bound, then give the two estimates for an individual support.

\subsubsection{A uniform bound for centered bits}

\begin{lemma}[Bounds for a centered bit]
\label{thm:qual-local-bounds}
For every map \(\tau\colon\Fp\to\{0,1\}\) over an odd prime field, the
function \(f=2(\tau-\E[\tau])\), with \(\mu=1-2\E[\tau]\), satisfies
\[
  \norm f_2=\sqrt{1-\mu^2},\qquad
  \norm f_{U^3}\le\kappa_0,\qquad
  \norm f_\infty\le1+|\mu|\le2 .
\]
\end{lemma}

\begin{proof}
Put \(V=1-\mu^2\). The values \(1+\mu\) and \(-(1-\mu)\) occur with
probabilities \((1-\mu)/2\) and \((1+\mu)/2\), respectively. This
proves the \(L^2\) and pointwise bounds, and gives
\[
  \E[f]=0,\qquad\E[f^2]=V,\qquad f^2=V+2\mu f .
\]

The \(U^3\) estimate uses the symmetry of the Fourier transform of a
real function. Specifically, every real \(q\) satisfies
\[
  \norm q_{U^2}^4
  \le(\E[q])^4+\frac12\bigl(\norm q_2^2-(\E[q])^2\bigr)^2
  \le\frac12\norm q_2^4+\frac12(\E[q])^4 .
\]
To see this, pair each nonzero frequency \(a\) with \(-a\); the two
Fourier magnitudes are equal, and the frequencies are distinct because
\(p\) is odd. If \(b_a=\abs{\widehat q(a)}^2\), with one representative
from each pair, then
\[
  \sum_{a\ne0}\abs{\widehat q(a)}^4
  =2\sum_{\text{pairs}}b_a^2
  \le2\Bigl(\sum_{\text{pairs}}b_a\Bigr)^2 .
\]
Adding the zero frequency and using Parseval gives the first inequality;
the second follows from \((\E[q])^2\le\norm q_2^2\).

Now let \(q_t(x)=f(x+t)f(x)\) and \(K(t)=\E_x[q_t(x)]\). Expanding
\(f(x+t)^2f(x)^2\) with \(f^2=V+2\mu f\) and using \(\E[f]=0\) gives
\[
  \norm{q_t}_2^2=V^2+4\mu^2K(t),\qquad\E_t[K(t)]=0 .
\]
Cauchy--Schwarz and the real-function bound applied to \(f\) give
\[
  \abs{K(t)}\le V,\qquad
  \E_t[K(t)^2]=\norm f_{U^2}^4\le\frac12V^2 .
\]
Applying the same bound to each \(q_t\), we obtain
\[
\begin{aligned}
  \norm f_{U^3}^8
  &=\E_t[\norm{q_t}_{U^2}^4]\\
  &\le\frac12\E_t[(V^2+4\mu^2K(t))^2]
       +\frac12\E_t[K(t)^4]\\
  &\le\frac12V^4+\Bigl(8\mu^4+\frac12V^2\Bigr)\E_t[K(t)^2]\\
  &\le\frac34V^4+4\mu^4V^2
   \le\frac34(V+\mu^2)^4=\frac34 .
\end{aligned}
\]
The second inequality uses \(\E_t[K(t)]=0\) and
\(K(t)^4\le V^2K(t)^2\); the last inequality follows by the
binomial theorem, since \(V,\mu^2\ge0\).
\end{proof}

\subsubsection{Two bounds for each support}

Write \(\mu_i=1-2\theta_i\). The preceding lemma gives two ways to make
the average small: bias decreases the \(L^2\) norms, whereas small bias
keeps the supremum norms close to one and lets the \(U^3\) bounds act.

\begin{lemma}[Bounds on \(\Gamma_C(\vec f_A)\)]
\label{lem:qual-padded-correlation}
Suppose \(2k-2\le n<3k-3\). For every \(A\subseteq[n]\) with
\(|A|\ge k\), put \(B=\sum_{i\in A}\mu_i^2\), \(e=n-2k+2\), and
\(q=q_3(n,k)=\floor{(3k-3-n)/2}\). Then
\[
  \abs{\Gamma_C(\vec f_A)}
  \le\min\Bigl\{
     \exp\bigl(q\log\kappa_0+\sqrt{nB}\bigr),\;
     \exp\bigl(e(\log2+1/2)-B/2\bigr)
   \Bigr\}.
\]
\end{lemma}

\begin{proof}
For the first bound, use all \(n\) evaluation forms defining \(C\),
with the constant function one outside \(A\); any \(k-1\) of the forms
are independent by the MDS property, and \(q\) is the number of
retained Gowers factors for \(t=3\). Since
\(0\le q\le(k-1)/2<|A|\), choose \(J\subseteq A\) of size \(q\). The
MDS von Neumann inequality (\cref{cor:iterated-cs-count}) with
\(t=3\), followed by \cref{thm:qual-local-bounds}, gives
\[
  |\Gamma_C(\vec f_A)|
  \le\prod_{j\in J}\|f_j\|_{U^3}\prod_{i\in A\setminus J}\|f_i\|_\infty
  \le\kappa_0^{\,q}\prod_{i\in A\setminus J}(1+|\mu_i|)
  \le\kappa_0^{\,q}\exp\Bigl(\sum_{i\in A}|\mu_i|\Bigr)
  \le\kappa_0^{\,q}\exp(\sqrt{nB}).
\]
If \(q=0\), the first inequality is just the supremum bound.

For the second bound, partition \([n]=I_1\sqcup I_2\sqcup E\) with
\(|I_1|=|I_2|=k-1\) and \(|E|=e\). Bound the factors on \(E\) by two
and apply Cauchy--Schwarz to the other two products. For
\(\vec X\sim U_C\),
\[
  |\Gamma_C(\vec f_A)|
  \le2^e\prod_{t=1}^2
     \Bigl(\E_{\vec X\sim U_C}
       \Bigl[\prod_{i\in A\cap I_t}|f_i(X_i)|^2\Bigr]\Bigr)^{1/2}
  =2^e\prod_{i\in A\setminus E}\sqrt{1-\mu_i^2},
\]
where the equality uses independence and uniformity within each
\(I_t\), again from the MDS property. Since
\(\sqrt{1-x}\le\exp(-x/2)\) and \(\sum_{i\in A\cap E}\mu_i^2\le e\),
\[
  2^e\prod_{i\in A\setminus E}\sqrt{1-\mu_i^2}
  \le\exp\Bigl(e\log2-\frac12\sum_{i\in A\setminus E}\mu_i^2\Bigr)
  \le\exp\bigl(e(\log2+1/2)-B/2\bigr).
  \qedhere
\]
\end{proof}

\begin{proof}[Proof of \cref{prop:qual-high-sum}]
We have \(r_0\ge k-1\), so \cref{lem:qual-padded-correlation} applies
to every support in the sum. If \(B\le\beta_0n\), use its first bound;
otherwise use its second. Since \(q=(3c-1)n/2+O(1)\) and
\(e=(1-2c)n\), in either case
\[
  |\Gamma_C(\vec f_A)|
  \le\exp\Bigl(n\max\Bigl\{
     \frac{3c-1}{2}\log\kappa_0+\sqrt{\beta_0},\;
     (1-2c)(\log2+1/2)-\frac{\beta_0}2
   \Bigr\}+O(1)\Bigr).
\]
There are at most
\(\sum_{0\le t<\varepsilon n}\binom nt\le(n+1)\exp(nh(\varepsilon))\)
such supports, since their complements have fewer than
\(\varepsilon n\) elements. Multiplying the two bounds proves the
proposition.
\end{proof}

\subsection{Combining the estimates}
\label{sec:splice}

Both estimates have strictly negative exponents at rate half when
the cutoff is sufficiently close to \(n\). We use continuity to keep
both exponents negative at a slightly smaller rate.

\begin{proof}[Proof of \cref{thm:grs-below-half}]
Consider the continuous formulas for \(G\) and \(H\) in
\cref{lem:elementary-low-band,prop:qual-high-sum}.
Since \(H(1/2,0)<0\), fix
\(0<\varepsilon<1/4\) with \(H(1/2,\varepsilon)<0\). For this
\(\varepsilon\), \(G(1/2,0,\varepsilon)=-\varepsilon^2/2<0\). By
continuity there is a fixed \(0<\delta<\min\{\varepsilon/3,1/6\}\)
such that, with \(c_0=1/2-\delta\) and \(\zetaGap=\delta\),
\[
  G(c_0,\zetaGap,\varepsilon)<0,\qquad H(c_0,\varepsilon)<0 .
\]
These choices also give \(1/3<c_0<1/2\), \(0<\zetaGap<c_0\), and
\(1-\varepsilon<1-3\delta=2c_0-\zetaGap\), as the two estimates
require. All these constants are independent of \(n\) and \(p\).

Put
\[
  k_n=\ceil{c_0n}+1,\qquad
  r_0=\floor{(1-\varepsilon)n},\qquad
  R_n=\floor{k_n-1-\zetaGap n}.
\]
We first prove the claim for a code \(C\) of dimension \(k_n-1\), and
write \(D=C^\perp\). Since \(k_n-1\ge\zetaGap n\), the consequence of
\cref{thm:rs-list-decoding} described in \cref{sec:list-decoding}
gives \(\List_D(R_n)\le L_n=n^{c_{\zetaGap}}\) for a constant
\(c_{\zetaGap}\) and all sufficiently large \(n\), uniformly over the
code parameters; as \(\zetaGap\) is fixed, \(\log L_n=O(\log n)=o(n)\).

All averages \(\Gamma_C(\vec g_A)\) on nonempty supports of size
smaller than \(k_n\) vanish by \cref{lem:binary-components}. For
\(k_n\le r\le r_0\), \cref{lem:elementary-low-band} and
\(G(c_0,\zetaGap,\varepsilon)<0\) give
\[
  \sum_{|A|=r}|\Gamma_C(\vec g_A)|^2\le\exp(-\beta_{\mathrm L}n)
\]
for a fixed \(\beta_{\mathrm L}>0\) and all sufficiently large \(n\),
uniformly in \(r\). For the remaining support sizes,
\(2(k_n-1)\le n<3(k_n-1)\) for large \(n\), and
\cref{prop:qual-high-sum} gives
\[
  \frac12\sum_{A\subseteq[n]\colon|A|>r_0}|\Gamma_C(\vec f_A)|
  \le\exp(-\beta_{\mathrm H}n)
\]
for a fixed \(\beta_{\mathrm H}>0\); here \((k_n-1)/n\to c_0\), so
continuity and \(H(c_0,\varepsilon)<0\) ensure that the rounding in
\(k_n\) does not change the sign of the exponent for large \(n\).

\Cref{prop:master-inequality} now yields
\[
  \SD\bigl(\vec\tau(\vec U_C),\vec\tau(\vec U_{\Fp^n})\bigr)
  \le\frac12\sqrt{n\exp(-\beta_{\mathrm L}n)}+\exp(-\beta_{\mathrm H}n)
  \le\exp(-\gamma n)
\]
for some \(\gamma>0\) and all sufficiently large \(n\); the factor
\(n\) counts the support sizes in the first sum and is absorbed by the
exponential decay. None of the constants depends on the leakage maps.
For a larger GRS code, the subcode obtained by restricting the
polynomial degree to less than \(k_n-1\) has the dimension just
treated, and \cref{lem:code-monotonicity} gives the same guarantee for
every larger dimension.

Choose any fixed \(c_{\mathrm{bh}}\) with \(c_0<c_{\mathrm{bh}}<1/2\).
The code theorem follows because every GRS code of rate at least
\(c_{\mathrm{bh}}\) has dimension at least \(\ceil{c_0n}\). If a
Shamir scheme has \(k\ge c_{\mathrm{bh}}n\), then for sufficiently
large \(n\), \(k-1\ge c_{\mathrm{bh}}n-1\ge c_0n\), so its code
dimension is also at least \(\ceil{c_0n}\), and
\cref{lem:shamir-from-code} proves \cref{thm:intro-main} with the same
\(c_{\mathrm{bh}}\) and leakage distance at most \(2e^{-\gamma n}\).
\end{proof}

\begin{corollary}[An explicit constant]
\label{cor:explicit-below-half}
The conclusions of \cref{thm:grs-below-half,thm:intro-main} hold with
\(c_{\mathrm{bh}}=1/2-10^{-14}\).
\end{corollary}

\begin{proof}
Taking
\(\varepsilon=6\cdot10^{-7}\), \(\zetaGap=2\cdot10^{-14}\), and
\(c_0=1/2-2\cdot10^{-14}\) gives
\(G(c_0,\zetaGap,\varepsilon)<-5.6\cdot10^{-14}\) and
\(H(c_0,\varepsilon)<-3.7\cdot10^{-7}\).
These choices satisfy \(0<\zetaGap<c_0<1/2\),
\(1-\varepsilon<2c_0-\zetaGap\), and
\(c_0<c_{\mathrm{bh}}<1/2\), so the preceding proof applies.
\end{proof}

\section*{Acknowledgments}
The author used ChatGPT (Sol and Astra) as interactive research
assistants throughout this work: to discuss and refine ideas and
arguments, to check calculations and parameter choices, to locate
related work, and to help draft and revise the text. Some mathematical
suggestions made by the models were incorporated after the author's
own verification. In particular, discussions with the models brought
recent list-decoding and list-recovery results
for Reed--Solomon codes to the author's attention and helped assess
their applicability. The author supplied
the research questions and the main ideas, evaluated and corrected the
models' output, and takes full responsibility for the correctness of
the results, the references, and the final presentation.

\begingroup
\small
\raggedright
\bibliographystyle{alpha}
\bibliography{references}
\endgroup

\clearpage
\appendix
\section{List recovery with growing input lists}
\label{app:growing-lists}

We prove \cref{thm:rs-list-decoding}, keeping the dependence on the
input-list size explicit. Interpolation gives a differential equation
satisfied by every candidate. We then count its solutions using their
initial coefficients and their agreement equations. We state the
interpolation input from~\cite{Jeronimo26} explicitly and give the
remaining counting argument below.

The Hasse derivatives are defined by
\(P(X+U)=\sum_{j\ge0}P^{[j]}(X)U^j\).
Write \(\deg_Y^{\mathrm{tot}}Q\) for total degree in all the \(Y\)-variables.

\begin{lemma}[Interpolation, from~{\cite[Proposition~3.7]{Jeronimo26}}]
\label{lem:app-interpolation}
Fix \(0<\zetaGap<1\). There are constants \(b_{\zetaGap}>0\) and
\(t_{\zetaGap}\ge1\) such that the following holds. Let
\(\alpha_1,\ldots,\alpha_n\in\Fp\) be distinct, let
\(S_1,\ldots,S_n\subseteq\Fp\), and put \(M=\sum_i|S_i|\).
If integers \(K,A,t\) satisfy
\[
  2\le K\le A\le n,\qquad t_{\zetaGap}\le t<K,\qquad
  \frac{A}{K-1}\ge\frac1{1-\zetaGap/2},\qquad
  M\le b_{\zetaGap}(K-1)t^{\zetaGap/3},
\]
then there is a nonzero \(Q(X,Y_0,\ldots,Y_t)\) with
\[
  \deg_X Q\le t^3n,\qquad
  \deg_Y^{\mathrm{tot}}Q\le
       \ceil{\frac{t^3A}{K-1}},
\]
such that every \(P\in\Fp[X]\) of degree below \(K\) agreeing with
the lists on at least \(A\) coordinates satisfies
\(Q(X,P,P^{[1]},\ldots,P^{[t]})=0\).
\end{lemma}

Apply the cited proposition with derivative order \(t\), multiplicity
\(t^3\), \(\vartheta=\zetaGap/2\), and \(\lambda=2/3\). Its incidence
exponent is \(\lambda\vartheta=\zetaGap/3\). The monomial restrictions
in~\cite[Equations~(11)--(12)]{Jeronimo26} give the stated total
\(Y\)-degree bound and \(\deg_XQ<t^3A\le t^3n\).
The interpolation works over any coefficient field; no upper bound on
\(p\) is required.

\begin{lemma}[Counting solutions with many agreements]
\label{lem:app-counting}
Let \(1\le h\le n\), \(t,B\ge1\), and \(D_0\ge0\) be integers, with
\(p>\max\{h-1,B\}\). Fix distinct \(\alpha_i\in\Fp\), lists \(S_i\)
with \(M=\sum_i|S_i|\ge1\), and a nonzero
\(Q\in\Fp[X,Y_0,\ldots,Y_t]\)
with \(\deg_X Q\le D_0\) and \(\deg_Y^{\mathrm{tot}}Q\le B\).
The number of polynomials \(P\in\Fp[X]\) of degree below \(h\) satisfying
\[
  Q(X,P,P^{[1]},\ldots,P^{[t]})=0,\qquad
  |\{i:P(\alpha_i)\in S_i\}|\ge h
\]
is at most
\[
  B^2\bigl(D_0+B(h-1)+1\bigr)
       \sum_{j=0}^t(ME)^j,\qquad E=8B(h+1)^2.
\]
\end{lemma}

\begin{proof}
We make the recurrence and counting arguments
of~\cite[Lemmas~5.6--5.7 and~6.2]{Jeronimo26} explicit.
If \(h\le t\), there are at most \(M^h\) ordered choices of \(h\)
agreements at distinct coordinates, each determining at most one
polynomial \(P\). This satisfies the claimed bound.
Assume \(h>t\), and work over the algebraic closure \(\overline{\Fp}\).

We use the following standard form of B\'ezout's inequality
(see~\cite[Lemma~5.3]{Jeronimo26}). The degree of an irreducible algebraic
set is the number of points, counted with multiplicity, in a generic
linear section of complementary dimension. A hypersurface defined by a
nonzero polynomial of degree at most \(B\) has sum of component degrees
at most \(B\). Intersecting an irreducible component with a degree-\(E\)
equation that does not vanish identically on it lowers its dimension
and increases the sum of degrees by a factor of at most \(E\).
A zero-dimensional algebraic set has at most its sum of degrees many points.

If \(Q\) is independent of all \(Y\)-variables, there are no solutions.
Otherwise let \(s\le t\) be its largest \(Y\)-index and put
\(H=\partial Q/\partial Y_s\ne0\); nonvanishing follows from \(B<p\).
First count candidates with
\(H(X,P,\ldots,P^{[s]})\ne0\). Fix a set of
\(D_0+B(h-1)+1\) distinct points \(\alpha\in\overline{\Fp}\),
independently of \(P\). Since this polynomial has degree at most
\(D_0+(B-1)(h-1)\), it is nonzero at one of these points.
The algebraic closure allows this choice even when the number of
points exceeds \(p\).

\paragraph{Recovering the coefficients at a fixed \(\alpha\).}
Write \(P(\alpha+U)=\sum_{q=0}^{h-1}c_qU^q\), with
\(c_q\in\overline{\Fp}\), put
\(\vec c=(c_0,\ldots,c_s)\), and set
\(H_\alpha(\vec c)=H(\alpha,c_0,\ldots,c_s)\).
The identity
\[
  P^{[i]}(\alpha+U)=
       \sum_{q=i}^{h-1}\binom qi c_qU^{q-i}
\]
shows that the coefficient of \(U^j\), for \(1\le j\le h-1-s\),
in the substituted equation is
\[
  \binom{s+j}{s}H_\alpha(\vec c)c_{s+j}
       +R_j(c_0,\ldots,c_{s+j-1})=0.
\]
Indeed, in \(P^{[s]}(\alpha+U)\), the coefficient \(c_{s+j}\)
first occurs as \(\binom{s+j}{s}c_{s+j}U^j\). In every lower
derivative \(P^{[i]}\) it occurs at order \(s+j-i>j\).
Thus it contributes to the coefficient of \(U^j\) only through a
single \(Y_s\) factor, with all other factors evaluated at \(U=0\);
the resulting multiplier is \(H_\alpha(\vec c)\).
Since \(s+j\le h-1<p\), the binomial coefficient is nonzero.
On \(H_\alpha\ne0\), these equations recursively give
\[
  c_{s+j}=\frac{N_j(\vec c)}{H_\alpha(\vec c)^{2j-1}},
  \qquad \deg N_j\le4Bj^2.
\]
Here and below, degrees are in the initial coefficients \(\vec c\).
To verify the bound, consider the noninitial factors
\(c_{s+a_1},\ldots,c_{s+a_v}\) in a monomial of \(R_j\).
Their orders satisfy \(\sum_u a_u\le j\), and the case
\(v=1,a_1=j\) is absent. Thus, for \(j\ge2\),
\[
  \sum_u(2a_u-1)\le2j-2,\qquad
  \sum_u a_u^2\le j^2-2j+2.
\]
Using the previous recurrences and clearing denominators with
\(H_\alpha^{2j-2}\) gives
\[
  \deg N_j
     \le B+4B(j^2-2j+2)+(2j-2)B
     \le4Bj^2.
\]
For \(j=1\), directly \(\deg N_1\le B\). This proves the bounds
without any hidden dependence on \(B\) or \(t\).

\paragraph{Counting the initial coefficients.}
The recurrences define a polynomial \(P_{\vec c}\) of degree below \(h\)
whenever \(H_\alpha(\vec c)\ne0\). Its first Taylor coefficients are
\(\vec c\), so this parametrization is injective.
Every candidate at this \(\alpha\) lies in
\[
  Q(\alpha,c_0,\ldots,c_s)=0,\qquad H_\alpha(\vec c)\ne0.
\]
We may count all such initial vectors over \(\overline{\Fp}\) satisfying
the agreement equations, even if their polynomials do not solve the full
differential equation.
When this set is nonempty, its defining polynomial is nonzero, because
its derivative in \(c_s\) is \(H_\alpha\). The hypersurface therefore
has dimension at most \(s\) and sum of component degrees at most \(B\).

Each of the \(M\) possible agreements becomes a polynomial equation
after multiplication by \(H_\alpha^{2h}\):
\[
  H_\alpha(\vec c)^{2h}\bigl(P_{\vec c}(\alpha_i)-b\bigr)=0,
       \qquad b\in S_i.
\]
The preceding numerator bounds give
\[
  \deg_{\vec c}
    \bigl[H_\alpha^{2h}(P_{\vec c}(\alpha_i)-b)\bigr]
     \le4B(h-1)^2+2hB+1\le E.
\]
Any \(h\) satisfied agreements at distinct coordinates determine
\(P_{\vec c}\), hence \(\vec c\), uniquely where \(H_\alpha\ne0\).
Let \(W\) be a positive-dimensional irreducible component through a
candidate \(\vec c\). Since \(H_\alpha(\vec c)\ne0\), the set
\(W\cap\{H_\alpha\ne0\}\) is a nonempty open subset of \(W\).
If all of the candidate's agreement equations vanished identically on
\(W\), every point in this open subset would give the same polynomial
\(P_{\vec c}\), hence the same initial vector \(\vec c\), a
contradiction. At least one agreement equation therefore cuts \(W\)
properly. Repeating lowers the dimension to zero after at most \(s\)
cuts. Hence every candidate at this \(\alpha\) lies in a zero-dimensional
component obtained from a component of \(\{Q(\alpha,\vec c)=0\}\) by
at most \(s\) successive proper cuts with agreement equations.
Components contained in \(H_\alpha=0\) can be discarded.
For each ordered choice of equations, retain only proper intersections
at each step. For \(j\) cuts there are at most \(M^j\) such choices,
and B\'ezout bounds their contribution by \(BE^j\) for each choice. Hence
\[
  \#\{\text{candidates at this }\alpha\}
       \le B\sum_{j=0}^s(ME)^j.
\]

The remaining candidates satisfy \(H(X,P,\ldots,P^{[s]})=0\)
identically. Repeat with \(H\) in place of \(Q\). Each repetition
lowers the total \(Y\)-degree without increasing the \(X\)-degree,
so there are at most \(B\) stages. At each stage there are at most
\(D_0+B(h-1)+1\) choices of \(\alpha\) from the fixed set. Summing the displayed
bound over these stages and points proves the lemma.
\end{proof}

\begin{proof}[Proof of \cref{thm:rs-list-decoding}]
Write \(h\) for the code dimension. Scaling \(S_i\) by \(v_i^{-1}\)
reduces GRS codes to ordinary RS codes. Fix lists of size at most
\(\ell\) and put \(M=\sum_i|S_i|\le n\ell\); if \(M=0\), there are
no candidates. Set
\[
  A=h+\ceil{\zetaGap n}\le n,\qquad
  K=\ceil{(1-\zetaGap/2)A}.
\]
The agreement threshold \(A\) corresponds exactly to the radius
\(\floor{n-h-\zetaGap n}\).
For sufficiently large \(n\), depending only on \(\zetaGap\)
(ensured by enlarging \(c_{\zetaGap}\)),
\[
  K-h\ge\zetaGap n/2,\qquad
  K-1\ge(1-\zetaGap/2)A/2\ge\zetaGap n/4,\qquad
  \frac{A}{K-1}>\frac1{1-\zetaGap/2}.
\]
Choose \(a_{\zetaGap}\ge t_{\zetaGap}\) with
\(b_{\zetaGap}\zetaGap a_{\zetaGap}^{\zetaGap/3}\ge4\), and put
\(t=\ceil{a_{\zetaGap}\ell^{3/\zetaGap}}\). Then
\[
  M\le n\ell\le b_{\zetaGap}(K-1)t^{\zetaGap/3}.
\]
Increasing \(c_{\zetaGap}\), the hypothesis
\(n\ge c_{\zetaGap}\ell^{9/\zetaGap}\) ensures \(t<K\) and
\(5t^3\le n\). Apply \cref{lem:app-interpolation}; every candidate
of degree below \(h\le K\) satisfies its equation, with
\[
  D_0=t^3n\le n^2,\qquad
  B=\ceil{\frac{t^3A}{K-1}}\le5t^3\le n<p.
\]
Thus \cref{lem:app-counting} applies. Since \(h\le n<p\),
\(M\le n^2\), and \(E=8B(h+1)^2\le32n^3\), it gives
\[
  \#\{\text{candidates}\}
     \le3n^4(t+1)(32n^5)^t
     \le n^{18t}
     \le n^{c_{\zetaGap}\ell^{3/\zetaGap}},
\]
after enlarging \(c_{\zetaGap}\) once more. All constants depend
only on \(\zetaGap\), and none depend on \(p\).
\end{proof}

\section{An alternative list-recovery calculation from DKT}
\label{app:dkt-lists}

We give a separate derivation of \cref{thm:rs-list-decoding} from
Dao, Kominers, and Thaler (DKT)~\cite{DKT26}.
Their Theorem~G.2 states a list-recovery bound for fixed input-list
size. The finite estimates in their Lemma~6.2 and the proof of
Theorem~6.7 make the dependence on that size explicit, as follows.

\begin{proof}[Alternative proof of \cref{thm:rs-list-decoding}]
Write \(h\) for the code dimension and fix lists \(S_i\) of size at
most \(\ell\). Scaling \(S_i\) by \(v_i^{-1}\) reduces GRS codes to
ordinary RS codes. We count polynomials of degree below \(h\) with
at least \(A=h+\ceil{\zetaGap n}\) agreements; this is exactly the
radius \(\floor{n-h-\zetaGap n}\) in the theorem.

\paragraph{Parameters.}
Set
\[
  \rho=\max\{h/n,\zetaGap/4\},\qquad
  a=\rho+\zetaGap/2.
\]
Since \(h\le(1-\zetaGap)n\),
\[
  \zetaGap/4\le\rho\le1-3\zetaGap/4,\qquad
  a\le1-\zetaGap/4<1,\qquad h\le\rho n,\qquad an\le A.
\]
Choose the derivative order \(t\), interpolation multiplicity \(s\),
and total degree bound \(B\) in the derivative variables by
\[
  t=\ceil{\max\left\{519,\frac16
          \left(\frac{320\ell}{9\zetaGap}\right)^{2/\zetaGap}
                     \right\}},\qquad
  s=\ceil{300t^2\log(6t)},\qquad
  B=\ceil{8s/\zetaGap}.
\]
Here \(t\) is denoted by \(d\) in DKT. Their Equation~(83) gives
\[
\begin{aligned}
  \Gamma(\rho,a,t)
    &=\frac{9\rho}{40}(6t)^{(a-\rho)/a}(1+1/t)\\
    &\ge\frac{9\zetaGap}{160}(6t)^{\zetaGap/2}(1+1/t)
     \ge2\ell(1+1/t).
\end{aligned}
\]
DKT's Lemma~6.2 bounds the finite interpolation ratio from below by
\(\Gamma(\rho,a,t)e^{-1677/10^6}>\ell\) when the local constraint
rank is nonzero; when it is zero, the same lemma gives a nonzero
interpolation space directly. Thus the interpolation
condition in the proof of their Theorem~G.2 holds for lists of size
\(\ell\), with the displayed finite choice of \(s\).

\paragraph{Length and characteristic.}
The interpolant in the proof of DKT's Theorem~6.7 has total degree
at most \(\ceil{2s/\rho}\le B\) in the derivative variables, and the
construction requires
\[
  n\ge\ceil{\max\left\{
     \frac{t+1}{\rho},\ \ceil{2s/\rho}+1,\ \frac1{1-\rho}
                    \right\}}.
\]
Our choices give \(t\ge\ell\), \(s\ge t+1\), and
\[
  \frac{t+1}{\rho}\le\frac{4(t+1)}{\zetaGap}\le B,\qquad
  \frac1{1-\rho}\le\frac4{3\zetaGap}<B,\qquad
  \ell\le B,\qquad 4/\zetaGap<B.
\]
Consequently \(n\ge B+2\) suffices. It also gives
\(p>n>\max\{h-1,t,B\}\), satisfying the characteristic
condition of DKT's Theorem~G.2. Finally,
\[
  t=O_{\zetaGap}(\ell^{2/\zetaGap}),\qquad
  B=O_{\zetaGap}(\ell^{4/\zetaGap}\log(2\ell)),
\]
so the theorem's hypothesis \(n\ge c_{\zetaGap}\ell^{9/\zetaGap}\)
ensures \(n\ge B+2\) after increasing \(c_{\zetaGap}\).

\paragraph{List size.}
DKT's Theorem~G.2, Equation~(147), now gives
\[
\begin{aligned}
  \#\{\text{candidate polynomials}\}
    &\le B^2\left(\frac{2\ell B}{a-\rho}\right)^t n^t
     =B^2\left(\frac{4\ell B}{\zetaGap}\right)^t n^t\\
    &\le n^{4t+2}\le n^{5t}
     \le n^{c_{\zetaGap}\ell^{3/\zetaGap}}.
\end{aligned}
\]
The second line uses \(B,\ell,4/\zetaGap\le n\) and \(t\ge2\);
the last inequality follows by increasing \(c_{\zetaGap}\) once
more. All constants depend only on \(\zetaGap\), not on \(p\).
\end{proof}

\end{document}